\documentclass[12pt]{article}

\pdfoutput=1
\usepackage[margin=1in]{geometry}
\usepackage[T1]{fontenc}
\usepackage{mathpazo}
\usepackage{eucal}
\usepackage{dsfont}
\usepackage{amsmath}
\usepackage{amssymb}
\usepackage{amsthm}
\usepackage{mathtools}
\usepackage{authblk}
\usepackage{tikz}
\usetikzlibrary{calc,shapes}

\usepackage{titlesec}
\usepackage{pifont}
\usepackage{pdfpages}
\usepackage{hyperref}
\hypersetup{
  pdfpagemode=UseNone,
  colorlinks=true,
  citecolor=blue,
  linkcolor=blue,
  urlcolor=blue
}

\newtheorem{theorem}{Theorem}
\newtheorem{lemma}[theorem]{Lemma}
\newtheorem{prop}[theorem]{Proposition}
\newtheorem{cor}[theorem]{Corollary}

\theoremstyle{definition}
\newtheorem{definition}[theorem]{Definition}
\newtheorem{remark}[theorem]{Remark}

\titleformat{\section}[hang]{\Large\bfseries\filright}{\thesection.}{.5em}{}
\titleformat{\subsection}[hang]{\large\bfseries\filright}{%
  \thesubsection.}{.5em}{}

\newcommand{\tsp}{\mspace{1mu}}
\newcommand{\htsp}{\mspace{0.5mu}}

\newcommand{\abs}[1]{\lvert #1 \rvert}
\newcommand{\bigabs}[1]{\bigl\lvert #1 \bigr\rvert}
\newcommand{\Bigabs}[1]{\Bigl\lvert #1 \Bigr\rvert}

\renewcommand\natural{\mathbb{N}}
\newcommand\integer{\mathbb{Z}}

\newcommand{\microspace}{\mspace{0.5mu}}
\newcommand{\op}[1]{\operatorname{#1}}
\newcommand{\tr}{\operatorname{Tr}}

\newcommand{\ket}[1]{\lvert\microspace #1 \microspace \rangle}
\newcommand{\bigket}[1]{\bigl\lvert\microspace #1 \microspace \bigr\rangle}

\newcommand{\bra}[1]{\langle\microspace #1 \microspace \rvert}
\newcommand{\bigbra}[1]{\bigl\langle\microspace #1 \microspace \bigr\rvert}

\newcommand\I{\mathds{1}}

\newcommand{\setft}[1]{\mathrm{#1}}
\newcommand{\Density}{\setft{D}}
\newcommand{\Pos}{\setft{Pos}}

\newenvironment{mylist}[1]{\begin{list}{}{
	\setlength{\leftmargin}{#1}
	\setlength{\rightmargin}{0mm}
	\setlength{\labelsep}{2mm}
	\setlength{\labelwidth}{8mm}
	\setlength{\itemsep}{0mm}}}
	{\end{list}}

\newcommand{\reg}[1]{\mathsf{#1}}
\newcommand{\class}[1]{\textup{#1}}

\newcommand\X{\mathcal{X}}
\newcommand\Y{\mathcal{Y}}

\newcommand\A{\mathcal{A}}
\newcommand\B{\mathcal{B}}

\newcommand\C{\mathcal{C}}

\renewcommand\P{\mathcal{P}}

\definecolor{White}{rgb}{1,1,1}
\definecolor{Black}{rgb}{0,0,0}
\definecolor{LightGray}{rgb}{.81,.81,.81}
\colorlet{ChannelColor}{LightGray}
\colorlet{ChannelTextColor}{Black}
\colorlet{ReadoutColor}{White}
 
\begin{document}

\title{Complexity limitations on one-turn\\ quantum refereed games}
\author{Soumik Ghosh}
\author{John Watrous}

\affil{
  Institute for Quantum Computing and School of Computer Science\protect\\
  University of Waterloo, Canada\vspace{2mm}
}

\renewcommand\Affilfont{\normalsize\itshape}
\renewcommand\Authfont{\large}
\setlength{\affilsep}{4mm}
\renewcommand\Authsep{\rule{10mm}{0mm}}
\renewcommand\Authand{\rule{10mm}{0mm}}

\date{February 4, 2020}

\maketitle

\begin{abstract}
  This paper studies complexity theoretic aspects of quantum refereed games,
  which are abstract games between two competing players that send quantum
  states to a referee, who performs an efficiently implementable joint
  measurement on the two states to determine which of the player wins.
  The complexity class $\class{QRG}(1)$ contains those decision problems for
  which one of the players can always win with high probability on
  yes-instances and the other player can always win with high probability on
  no-instances, regardless of the opposing player's strategy.
  This class trivially contains $\class{QMA} \cup \class{co-QMA}$ and is known
  to be contained in $\class{PSPACE}$.
  We prove stronger containments on two restricted variants of this class.
  Specifically, if one of the players is limited to sending a classical
  (probabilistic) state rather than a quantum state, the resulting complexity
  class $\class{CQRG}(1)$ is contained in $\exists\cdot\class{PP}$
  (the nondeterministic polynomial-time operator applied to $\class{PP}$);
  while if both players send quantum states but the referee is forced to
  measure one of the states first, and incorporates the classical outcome of
  this measurement into a measurement of the second state, the resulting class
  $\class{MQRG}(1)$ is contained in $\class{P}\cdot\class{PP}$
  (the unbounded-error probabilistic polynomial-time operator applied to
  $\class{PP}$).
\end{abstract}

\section{Introduction}

Abstract notions of games have long played an important role in complexity
theory.
For example, combinatorial games provide complete problems for various
complexity classes \cite{DemaineH09}, the notion of alternation is naturally
described in game-theoretic terms \cite{ChandraKS1981}, and interactive proof
systems \cite{Babai1985,BabaiM1988,GoldwasserMR1985,GoldwasserMR1989} and many
variants of them are naturally formulated as games
\cite{Condon1987,FeigenbaumKS1995}.

This paper is concerned with games between two competing, computationally
unbounded players, administered by a computationally bounded referee.
In the classical setting, complexity theoretic aspects of games of this form
were investigated in the 1990s by Koller and Megiddo \cite{KollerM1992},
Feigenbaum, Koller, and Shor \cite{FeigenbaumKS1995}, Condon, Feigenbaum, Lund,
and Shor \cite{CondonFLS1995,CondonFLS1997}, and Feige and Kilian
\cite{FeigeK1997}.
Quantum computational analogues of these games were later considered in
\cite{GutoskiW2005}, \cite{Gutoski2005}, \cite{GutoskiW2007}, and
\cite{JainW2009}.

Our focus will be on \emph{one-turn refereed games}, in which the players and
the referee first receive a common input string, and then each player sends a
single polynomial-length (quantum or classical) message to the referee, who
then decides which player has won.
We will refer to the two competing players as \emph{Alice} and \emph{Bob} for
convenience.
In the classical case Alice and Bob's messages may in general be described by
probability distributions over strings, while in the quantum case Alice and
Bob's messages are described by mixed quantum states, which are represented by
density operators.
In both cases, the referee's decision process must be specified by a
polynomial-time generated family of (quantum or classical) circuits.
Two complexity classes are defined---$\class{RG}(1)$ in the classical
setting\footnote{%
  We note explicitly that this nomenclature clashes with \cite{FeigeK1997},
  which defines $\class{RG}(1)$ in terms of \emph{one-round} (i.e., two-turn)
  refereed games, which is $\class{RG}(2)$ with respect to our naming
  conventions.}
and $\class{QRG}(1)$ in the quantum setting---consisting of all promise
problems $A = (A_{\text{yes}},A_{\text{no}})$ for which there exists a game
(either classical or quantum, respectively) such that Alice can win with high
probability on inputs $x\in A_{\text{yes}}$ and Bob can win with high
probability on inputs $x\in A_{\text{no}}$, regardless of the other player's
behavior.

In essence, the complexity classes $\class{RG}(1)$ and $\class{QRG}(1)$ may be
viewed as extensions of the classes $\class{MA}$ and $\class{QMA}$ in which
\emph{two competing Merlins}, one whose aim is to convince the referee (whose
role is analogous to Arthur, also called the verifier, in the case of 
$\class{MA}$ and $\class{QMA}$) that the input string is a yes-instance of
a given problem, and the other whose aim is to convince the referee that the
input string is a no-instance.

It is known that the complexity class $\class{RG}(1)$ is equal to
$\class{S}_2^p$, which refers to the second level of the
\emph{symmetric polynomial-time hierarchy} introduced by Canetti
\cite{Canetti1996} and Russell and Sundaram \cite{RussellS1998}.
This class is most typically defined in terms of quantifiers that suggest
games in which Alice and Bob choose polynomial-length strings (as opposed to
probability distributions of strings) to send to the referee, but the class
does not change if one adopts a bounded-error definition in which Alice and Bob
are allowed to make use of randomness \cite{FortnowIKU2008}.
Moreover, the class does not change if the referee is permitted the use of
randomness, again assuming a bounded-error definition.
An essential fact through which these equivalence may be proved,
due to Alth\"ofer \cite{Althofer1994} and Lipton and Young \cite{LiptonY1994},
is that non-interactive randomized games always admit near-optimal strategies
that are uniform over polynomial-size sets of strings.
It is also known that $\class{RG}(1)$ is closed under Cook reductions
\cite{RussellS1998} and satisfies
$\class{RG}(1) \subseteq \class{ZPP}^{\class{NP}}$ \cite{Cai2007}.

In contrast to the containment
$\class{RG}(1) \subseteq \class{ZPP}^{\class{NP}}$, the best upper-bound known
for $\class{QRG}(1)$ is that this class is contained in $\class{PSPACE}$
\cite{JainW2009}.
It is reasonable to conjecture that a stronger upper-bound on $\class{QRG}(1)$
can be proved.
Indeed, Gutoski and Wu \cite{GutoskiW2013} proved that
$\class{QRG}(2) = \class{PSPACE}$, where $\class{QRG}(2)$ is a two-turn
analogue of $\class{QRG}(1)$, in which the referee first sends
polynomial-length quantum messages to Alice and Bob, then receives responses
from them, and finally decides which player wins.
The classical analogue of $\class{QRG}(2)$, which we denote by $\class{RG}(2)$,
is also known to be equal to $\class{PSPACE}$ \cite{FeigeK1997}.

In this work we consider two restricted variants of $\class{QRG}(1)$, and
prove stronger upper-bounds than $\class{PSPACE}$ on these restricted
variants.
The first variant is one in which Alice is limited to sending a classical
message to the referee, while Bob is free to send a quantum state.
The resulting class, which we call $\class{CQRG}(1)$, is proved to be
contained in $\exists\cdot\class{PP}$ (the class obtained when the
nondeterministic polynomial-time operator is applied to $\class{PP}$).
This containment follows from an application of the
Alth\"ofer--Lipton--Young technique mentioned above, although in the quantum
setting the proof requires relatively recent tail bounds on sums of
matrix-valued random variables, as opposed to a more standard
Hoeffding--Chernoff type of bound that suffices in the classical case.
In particular, we make use of a tail bound of this sort due to Tropp
\cite{Tropp2012}.
The second variant we consider is one in which both Alice and Bob are free to
send quantum states, but where the referee must first measure Alice's state and
then incorporate the classical outcome of this measurement into a measurement
of Bob's state.
We call the corresponding class $\class{MQRG}(1)$, and prove the containment
$\class{MQRG}(1)\subseteq\class{P}\cdot\class{PP}$
(the class obtained when the unbounded error probabilistic polynomial-time
operator is applied to $\class{PP}$).

\section{Preliminaries}
\label{sec:preliminaries}

We assume the reader is familiar with basic aspects of computational complexity
theory and quantum information and computation.
There are four subsections included in this preliminaries section, the first of
which clarifies a few specific concepts, conventions, and definitions
concerning complexity theory.
The second subsection is concerned specifically with counting complexity, and
presents a development of some results on this topic that are central to this
paper.
Proofs are included because these results represent minor generalizations of
known results on counting complexity.
The third subsection discusses a few specific definitions and concepts from
quantum information and computation, along with a proof of a fact that may
be considered a known result, but for which a complete proof does not appear
in published form.
The final subsection states the tail bound due to Tropp mentioned above.

\subsection*{Complexity theory basics}

Throughout this paper, languages, promise problems, and functions on strings
are assumed to be over the binary alphabet $\Sigma = \{0,1\}$.
The set of natural numbers, including~0, is denoted $\mathbb{N}$.

A function of the form $p:\natural\rightarrow\natural$ is said to be
\emph{polynomially bounded} if there exists a deterministic Turing machine
that runs in polynomial time and outputs $0^{p(n)}$ on input $0^n$ for all
$n\in\natural$.
Unless it is explicitly indicated otherwise, the input of a given
polynomially bounded function $p$ is assumed to be the natural number
$\abs{x}$, for whatever input string $x\in\Sigma^{\ast}$ is being considered at
that moment.
With this understanding in mind, we will write $p$ in place of $p(\abs{x})$
when referring to the natural number output that is determined in this way.
For example, in Definition~\ref{def:NP and PP operators} below, all of the
occurrences of $p$ in the displayed equations are short for $p(\abs{x})$.
This convention helps to make formulas and equations more clear and less
cluttered.

A promise problem is a pair $A = (A_{\text{yes}},A_{\text{no}})$ of sets
of strings $A_{\text{yes}},A_{\text{no}}\subseteq\Sigma^{\ast}$ with
$A_{\text{yes}} \cap A_{\text{no}} = \varnothing$.
Strings in $A_{\text{yes}}$ represent yes-instances of a decision problem,
strings in $A_{\text{no}}$ represent no-instances, and all other strings
represent ``don't care'' inputs for which no restrictions are placed on
a hypothetical computation for that problem.

We fix a pairing function that efficiently encodes two strings
$x,y\in\Sigma^{\ast}$ into a single binary string denoted
$\langle x , y \rangle \in\Sigma^{\ast}$, and we assume
that this function satisfies some simple properties:
\begin{mylist}{8mm}
\item[1.]
  The length of the pair $\langle x , y \rangle$ depends only on the lengths
  $\abs{x}$ and $\abs{y}$, and is polynomial in these lengths.
\item[2.]
  The computation of $x$ and $y$ from $\langle x,y\rangle$, as well as the
  computation of $\langle x,y\rangle$ from $x$ and $y$, can be performed
  deterministically in polynomial time.
\end{mylist}
One suitable choice for such a function is suggested by the equation
\begin{equation}
  \langle a_1 a_2 \cdots a_n, b_1 b_2 \cdots b_m
  \rangle
  = 0\htsp a_1\htsp 0\htsp a_2\htsp \cdots \htsp 0\htsp a_n \htsp 1
  \htsp b_1\htsp b_2\htsp \cdots\htsp b_m
\end{equation}
for $a_1, a_2,\ldots,a_n,b_1,b_2,\ldots,b_m\in\Sigma$.
Any such pairing function may be extended recursively to obtain a tuple
function for any fixed number of inputs by taking
\begin{equation}
  \langle x_1,x_2,x_3,\ldots,x_k\rangle =
  \langle \langle x_1,x_2\rangle, x_3, \ldots, x_k\rangle
\end{equation}
for strings $x_1,\ldots,x_k\in\Sigma^{\ast}$, where $k\geq 3$.
Hereafter, when we refer to the computation of any function taking
multiple string-valued arguments, we assume that these input strings have been
encoded into a single string using this tuple function.
For instance, when $f$ is a function that represents a computation,
we write $f(x,y,z)$ rather than $f(\langle x,y,z\rangle)$.

Finally, we define the nondeterministic and probabilistic polynomial-time
operators, which may be applied to an arbitrary complexity class, as follows.

\begin{definition}
  \label{def:NP and PP operators}
  For a given complexity class of languages $\C$, the complexity classes
  $\exists\cdot\C$ and $\class{P}\cdot\C$ are defined as follows.
  \begin{mylist}{8mm}
  \item[1.]
    The complexity class $\exists\cdot\C$ contains all promise problems
    $A = (A_{\text{yes}},A_{\text{no}})$ for which there exists a language
    $B\in\C$ and a polynomially bounded function $p$ such that these two
    implications hold:
    \begin{equation}
      \begin{aligned}
        x\in A_{\text{yes}} & \Rightarrow
        \Bigl\{y\in\Sigma^{p}\,:\,\langle x,y\rangle\in B\Bigr\}
        \not=\varnothing,\\
        x\in A_{\text{no}} & \Rightarrow
        \Bigl\{y\in\Sigma^{p}\,:\,\langle x,y\rangle\in B\Bigr\}
        =\varnothing.        
      \end{aligned}
    \end{equation}
  \item[2.]
    The complexity class $\class{P}\cdot\C$ contains all promise problems
    $A = (A_{\text{yes}},A_{\text{no}})$ for which there exists a language
    $B\in\C$ and a polynomially bounded function $p$ such that these two
    implications hold:
    \begin{equation}
      \begin{aligned}
        x\in A_{\text{yes}} & \Rightarrow
        \Bigabs{\Bigl\{y\in\Sigma^p\,:\,\langle x,y\rangle\in B\Bigr\}}
        > \frac{1}{2} \cdot 2^p,\\[2mm]
        x\in A_{\text{no}} & \Rightarrow
        \Bigabs{\Bigl\{y\in\Sigma^p\,:\,\langle x,y\rangle\in B\Bigr\}}
        \leq \frac{1}{2} \cdot 2^p.
      \end{aligned}
    \end{equation}
  \end{mylist}
\end{definition}

\subsection*{Counting complexity}

Counting complexity is principally concerned with the number of solutions
to certain computational problems.
Readers interested in learning more about counting complexity and some of its
applications are referred to the survey paper of Fortnow \cite{Fortnow1997}.
As was suggested at the beginning of the current section, we will require
some basic results on counting complexity that represent minor generalizations
of known results.
We begin with the following definition.

\begin{definition}
  \label{definition:Gap.C}
  Let $\C$ be any complexity class of languages over the alphabet $\Sigma$.
  A function $f:\Sigma^{\ast} \rightarrow \integer$ is a
  $\class{Gap}\cdot\C$ function if there exist languages $A,B\in\C$ and
  a polynomially bounded function $p$ such that
  \begin{equation}
    f(x) = \bigabs{\bigl\{ y\in\Sigma^p \,:\, \langle x,y\rangle \in A\bigr\}}
    -\bigabs{\bigl\{y\in\Sigma^p \,:\, \langle x,y\rangle \in B\bigr\}}      
  \end{equation}
  for all $x\in\Sigma^{\ast}$.
\end{definition}

We observe that this definition is slightly non-standard, as gap functions
are usually defined in terms of differences between the number
of accepting and rejecting computations of nondeterministic machines
(as opposed to a difference involving two potentially unrelated languages $A$
and $B$).
It is also typical that one focuses on specific choices for~$\C$,
particularly $\C = \class{P}$.
Our definition is, however, equivalent to the traditional definition in this
case, and we will write $\class{GapP}$ rather than $\class{Gap}\cdot\class{P}$
so as to be consistent with the standard name for this class of functions.
We will also be interested in the case $\C = \class{PP}$, which yields a class
of functions $\class{Gap}\cdot\class{PP}$ that is less commonly considered.

The following proposition is immediate from the definitions of
$\class{Gap}\cdot\C$ and $\class{P}\cdot\C$.

\begin{prop}
  \label{prop:Gap.C versus P.C}
  Let $\C$ be a complexity class of languages that is closed under
  complementation.
  A promise problem $A = (A_{\text{yes}},A_{\text{no}})$ is contained in
  $\class{P}\cdot\C$ if and only if there exists a $\class{Gap}\cdot\C$
  function $f$ such that
  \begin{equation}
    \begin{aligned}
      x\in A_{\text{yes}} & \Rightarrow f(x) > 0,\\
      x\in A_{\text{no}} & \Rightarrow f(x) \leq 0.
    \end{aligned}
  \end{equation}
\end{prop}

A key feature of the class of $\class{GapP}$ functions that facilitates its
use is that it possess strong closure properties.
This is true more generally for the class $\class{Gap}\cdot\C$ provided that
$\C$ itself possesses certain properties.
For the closure properties we require, it suffices that $\C$ is
nontrivial (meaning that $\C$ contains at least one language that is not
equal to $\varnothing$ or $\Sigma^{\ast}$) and is closed under the join
operation as well as polynomial-time truth-table reductions.
(The \emph{join} of languages $A$ and $B$ is defined as
$\{0x\,:\,x\in A\}\cup\{1x\,:\,x\in B\}$.)
These properties are, of course, possessed by both $\class{P}$ and
$\class{PP}$, with the closure of $\class{PP}$ under truth table reductions
having been proved by Fortnow and Reingold \cite{FortnowR1996} based on
methods developed by Beigel, Reingold, and Spielman \cite{BeigelRS1995}.

The lemmas that follow establish the specific closure properties we require.
For the first property the assumption that $\C$ is closed under joins and
polynomial-time truth-table reductions is not required; closure under Karp
reductions suffices.

\begin{lemma}
  \label{lemma:Gap.C closed under exponential sums}
  Let $\C$ be a nontrivial complexity class of languages that is closed under
  Karp reductions.
  Let $f\in\class{Gap}\cdot\C$ and let $p$ be a polynomially bounded function.
  The function
  \begin{equation}
    g(x) = \sum_{y\in\Sigma^p} f(x,y)
  \end{equation}
  is a $\class{Gap}\cdot\C$ function.
\end{lemma}

\begin{proof}
  By the assumption that $f\in\class{Gap}\cdot\C$, there exists a polynomially
  bounded function $q$ and languages $A_0,A_1\in\C$ such that
  \begin{equation}
    f(x,y) = \bigabs{\bigl\{
      z\in\Sigma^{q(\abs{\langle x,y\rangle})} \,:\,
      \langle x,y,z\rangle \in A_0\bigr\}}
    -\bigabs{\bigl\{
      z\in\Sigma^{q(\abs{\langle x,y\rangle})} \,:\,
      \langle x,y,z\rangle \in A_1\bigr\}}
  \end{equation}
  for all $x\in\Sigma^{\ast}$ and $y\in\Sigma^p$.
  By the assumptions on our pairing function described above, it is the case
  that $\abs{\langle x,y\rangle}$ depends only on $\abs{x}$ and $\abs{y}$, and
  therefore there exists a (necessarily polynomially bounded) function $r$ such
  that $r(\abs{x}) = p(\abs{x}) + q(\abs{\langle x,y\rangle})$ for all
  $x\in\Sigma^{\ast}$ and $y\in\Sigma^p$.
  Define
  \begin{equation}
    \begin{aligned}
      B_0 & = \bigl\{ \langle x,yz\rangle\,:\,
      y \in \Sigma^p,\; z \in \Sigma^{q(\abs{\langle x,y\rangle})},\;
      \langle x,y,z\rangle \in A_0\bigr\},\\[1mm]
      B_1 & = \bigl\{ (x,yz)\,:\,
      y \in \Sigma^p,\; z \in \Sigma^{q(\abs{\langle x,y\rangle})},\;
      \langle x,y,z\rangle \in A_1\bigr\}.
    \end{aligned}
  \end{equation}
  By the nontriviality and closure of $\C$ under Karp reductions, it is evident
  that $B_0,B_1\in\C$.
  It may be verified that
  \begin{equation}
    g(x) =
    \bigabs{\bigl\{w\in\Sigma^r \,:\,\langle x,w\rangle \in B_0\bigr\}}
    - \bigabs{\bigl\{w\in\Sigma^r \,:\,\langle x,w\rangle \in B_1\bigr\}}
  \end{equation}
  for all $x\in\Sigma^{\ast}$, and therefore $g\in\class{Gap}\cdot\C$.
\end{proof}


For the next lemma, and elsewhere in the paper, we will use the following
notation for convenience: $\Sigma^n_1$ denotes the set of all strings over the
binary alphabet $\Sigma$ that have length $n$ and contain exactly one
occurrence of the symbol 1.
It is therefore the case that $\abs{\Sigma^n_1} = n$.

\begin{lemma}
  \label{lemma:Gap.C closed under polynomial products}
  Let $\C$ be a nontrivial complexity class of languages that is closed under
  joins and polynomial-time truth table reductions.
  Let $f\in\class{Gap}\cdot\C$ and let $p$ be a polynomially bounded function.
  The function
  \begin{equation}
    g(x) = \prod_{y\in\Sigma_1^p} f(x,y)
  \end{equation}
  is a $\class{Gap}\cdot\C$ function.
\end{lemma}

\begin{proof}
  Given that $f\in\class{Gap}\cdot\C$, there exists a polynomially bounded
  function $q$ and languages $A_0,A_1\in\C$ such that
  \begin{equation}
    \label{eq:f(x,y)}
    f(x,y) = \Bigabs{\Bigl\{z \in \Sigma^{q(\abs{(x,y)})}\,:\,
      \langle x,y,z \rangle\in A_0\Bigr\}}
    - \Bigabs{\Bigl\{z\in\Sigma^{q(\abs{(x,y)})}\,:\,
      \langle x,y,z \rangle \in A_1\Bigr\}}
  \end{equation}
  for all $x,y\in\Sigma^{\ast}$.
  We may assume further that $A_0$ and $A_1$ are disjoint languages, for if
  they are not, we may replace $A_0$ and $A_1$ with $A_0 \cap \overline{A_1}$
  and $A_1 \cap \overline{A_0}$, respectively;
  this does not change the value of the right-hand side of the equation
  \eqref{eq:f(x,y)}, and the languages $A_0 \cap \overline{A_1}$ and
  $A_1 \cap \overline{A_0}$ must both be contained in $\C$ for
  $A_0,A_1\in\C$ by the closure of $\C$ under joins and truth-table
  reductions.
  
  By the assumptions on our pairing function described above, there exists a
  polynomially bounded function $r$ such that $r(\abs{x}) = q(\abs{(x,y)})$ for
  all $x\in\Sigma^{\ast}$ and $y\in\Sigma^p$.
  We will write $y_1,\ldots,y_p$ to denote the elements of $\Sigma^p_1$ sorted
  in lexicographic order.
  Define two languages $B_0$ and $B_1$ as follows:
  \begin{mylist}{8mm}
  \item[$\bullet$]
    $B_0$ is the language of all pairs $\langle x,z_1\cdots z_p\rangle$, where
    $x\in\Sigma^{\ast}$ and $z_1,\ldots,z_p\in\Sigma^r$, for which there exists
    a string $w\in\Sigma^p$ having even parity such that
    \begin{equation}
      \langle x,y_1,z_1\rangle\in A_{w_1},\ldots,\langle x,y_p,z_p\rangle
      \in A_{w_p}.
    \end{equation}
    
  \item[$\bullet$]
    $B_1$ is the language of all pairs $\langle x,z_1\cdots z_p\rangle$, where
    $x\in\Sigma^{\ast}$ and $z_1,\ldots,z_p\in\Sigma^r$, for which there exists
    a string $w\in\Sigma^p$ having odd parity such that
    \begin{equation}
      \langle x,y_1,z_1\rangle\in A_{w_1},\ldots,\langle x,y_p,z_p\rangle
      \in A_{w_p}.
    \end{equation}
  \end{mylist}
  Given that $A_0$ and $A_1$ are disjoint and contained in $\C$, along with the
  fact that $\C$ is closed under joins and truth-table reductions, it follows
  that $B_0,B_1\in\C$.
  The lemma now follows from the observation that
  \begin{equation}
    g(x) = \Bigabs{\Bigl\{z \in \Sigma^s\,:\,\langle x,z\rangle\in B_0
      \Bigr\}}
    - \Bigabs{\Bigl\{z \in \Sigma^s\,:\,\langle x,z\rangle\in B_1\Bigr\}}
  \end{equation}
  for all $x\in\Sigma^{\ast}$, where $s= p \cdot r$.
\end{proof}

\begin{lemma}
  \label{lemma:Gap.C matrix multiplication}
  Let $\C$ be a nontrivial complexity class of languages that is closed under
  joins and polynomial-time truth table reductions, let
  $f_0,f_1\in\class{Gap}\cdot\C$, and let $p$ and $q$ be polynomially bounded
  functions.
  For every string $x\in\Sigma^{\ast}$ and $y\in\Sigma^q_1$, define the matrix
  $M_{x,y}$ as
  \begin{equation}
    \begin{aligned}
      \op{Re}\bigl(\bra{z}M_{x,y}\ket{w}\bigr) & = f_0(x,y,z,w),\\
      \op{Im}\bigl(\bra{z}M_{x,y}\ket{w}\bigr) & = f_1(x,y,z,w),
    \end{aligned}
  \end{equation}
  for all $z,w\in\Sigma^p$.
  There exist $\class{Gap}\cdot\C$ functions $g_0$ and $g_1$ satisfying
  \begin{equation}
    \label{eq:g_0 and g_1}
    \begin{aligned}
      \op{Re}\bigl(\bra{z}M_{x,y_1}\cdots M_{x,y_q}\ket{w}\bigr)
      & = g_0(x,z,w),\\
      \op{Im}\bigl(\bra{z}M_{x,y_1}\cdots M_{x,y_q}\ket{w}\bigr)
      & = g_1(x,z,w),
    \end{aligned}
  \end{equation}
  for all $x\in\Sigma^{\ast}$ and $z,w\in\Sigma^p$, where $y_1,\ldots,y_q$
  denote the elements of $\Sigma^q_1$ sorted in lexicographic order.
\end{lemma}

\begin{proof}
  By the assumptions on $\C$ stated in the lemma, there must exist a
  $\class{Gap}\cdot\C$ function $h$ satisfying
  \begin{equation}
    \begin{aligned}
      h(x,y,0z,0w) & = f_0(x,y,z,w),\\
      h(x,y,0z,1w) & = f_1(x,y,z,w),\\
      h(x,y,1z,0w) & = -f_1(x,y,z,w),\\
      h(x,y,1z,1w) & = f_0(x,y,z,w),\\
    \end{aligned}
  \end{equation}
  for all $x\in\Sigma^{\ast}$, $y\in\Sigma^q_1$, and $z,w\in\Sigma^p$.
  The matrix $N_{x,y}$ defined as
  \begin{equation}
    \bra{u} N_{x,y} \ket{v} = h(x,y,u,v)
  \end{equation}
  for all $x\in\Sigma^{\ast}$, $y\in\Sigma^q_1$ and $u,v\in\Sigma^{p+1}$ may be
  visualized as a $2\times 2$ block matrix:
  \begin{equation}
    N_{x,y} =
    \begin{pmatrix}
      \op{Re}(M_{x,y}) & \op{Im}(M_{x,y})\\[1mm]
      -\op{Im}(M_{x,y}) & \op{Re}(M_{x,y})
    \end{pmatrix}.
  \end{equation}
  We observe that
  \begin{equation}
    \label{eq:block products}
    N_{x,y_1} \cdots N_{x,y_q} =
    \begin{pmatrix}
      \op{Re}\bigl(M_{x,y_1} \cdots M_{x,y_q}\bigl) &
      \op{Im}\bigl(M_{x,y_1} \cdots M_{x,y_q}\bigl)\\[1mm]
      -\op{Im}\bigl(M_{x,y_1} \cdots M_{x,y_q}\bigl) &
      \op{Re}\bigl(M_{x,y_1} \cdots M_{x,y_q}\bigl)
    \end{pmatrix}.
  \end{equation}
  Given that $h$ is a $\class{Gap}\cdot\C$ function, there must exist a
  $\class{Gap}\cdot\C$ function $F$ for which
  \begin{equation}
    F(x,u_0\cdots u_q,y_k) = h(x,y_k,u_{k-1},u_k)
  \end{equation}
  for all $x\in\Sigma^{\ast}$, $u_0,\ldots,u_q\in\Sigma^{p+1}$, and
  $k\in\{1,\ldots,q\}$.

  Finally, define
  \begin{equation}
    G(x,u_0 \cdots u_q) = \prod_{y\in\Sigma^q_1} F(x,u_0\cdots u_q,y)
    = h(x,y_1,u_0,u_1) \cdots h(x,y_q,u_{q-1},u_q)    
  \end{equation}
  for all $x\in\Sigma^{\ast}$ and $u_0,\ldots,u_q\in\Sigma^{p+1}$,
  as well as
  \begin{equation}
    \begin{aligned}
      g_0(x,z,w) & = \sum_{u\in\Sigma^{(q-1)(p+1)}} G(x,0zu0w), \\
      g_1(x,z,w) & = \sum_{u\in\Sigma^{(q-1)(p+1)}} G(x,0zu1w),
    \end{aligned}
  \end{equation}
  for all $x\in\Sigma^{\ast}$ and $z,w\in\Sigma^p$.
  It follows by Lemmas~\ref{lemma:Gap.C closed under exponential sums} and
  \ref{lemma:Gap.C closed under polynomial products} that
  $g_0,g_1\in\class{Gap}\cdot\C$.

  Observing that $g_0$ and $g_1$ satisfy the equations \eqref{eq:g_0 and g_1},
  which is perhaps most evident from the equation \eqref{eq:block products},
  the proof of the lemma is complete.
\end{proof}

\pagebreak[3]

\subsection*{Quantum information and quantum circuits}

The notation we use when discussing quantum information is standard for
the subject, and we refer the reader to the books
\cite{NielsenC2000,KitaevSV2002,Wilde2017,Watrous2018} for further details.
A couple of points concerning quantum information notation and conventions that
may be helpful to some readers follow.

First, when we refer to a \emph{register} $\reg{X}$, we mean a collection of
qubits that we wish to view as a single entity, and we then use the same
letter $\X$ in a scripted font to denote the finite-dimensional complex Hilbert
space associated with $\reg{X}$ (i.e., the space of complex vectors having
entries indexed by binary strings of length equal to the number of qubits in
$\reg{X}$).
The set of density operators acting on such a space is denoted $\Density(\X)$.

Second, a \emph{channel} transforming a register $\reg{X}$ into a register
$\reg{Y}$ is a completely positive and trace-preserving linear map $\Phi$ that
transforms each density operator $\rho\in\Density(\X)$ into a density operator
$\Phi(\rho)\in\Density(\Y)$.
(More generally, such a mapping $\Phi$ transforms arbitrary linear
operators acting on $\X$ into linear operators acting on $\Y$.)
The \emph{adjoint} of such a channel $\Phi$ is the uniquely determined linear
map $\Phi^{\ast}$ transforming linear operators acting on~$\Y$ into linear
operators acting on $\X$ that satisfies the equation
\begin{equation}
  \tr\bigl(P \tsp \Phi(\rho)\bigr) =
  \tr\bigl(\Phi^{\ast}(P)\tsp \rho\bigr)
\end{equation}
for all density operators $\rho\in\Density(\X)$ and all positive semidefinite
operators $P$ acting on~$\Y$.
The adjoint $\Phi^{\ast}$ of a channel $\Phi$ is not necessarily itself a
channel, but rather is a completely positive and \emph{unital} linear map,
which means that $\Phi^{\ast}(\I_{\Y}) = \I_{\X}$ (for $\I_{\X}$ and $\I_{\Y}$
denoting the identity operators acting on $\X$ and $\Y$, respectively).
Intuitively speaking, if $P$ is a measurement operator in the equation above,
one can think of $\Phi^{\ast}$ as transforming the measurement operator
$P$ into a new measurement operator $\Phi^{\ast}(P)$, with the probability of
this outcome for the state $\rho$ being the same as if one first applied $\Phi$
to $\rho$ and then measured with respect to $P$.

Now we will move on to \emph{quantum circuits}, which are acyclic networks of
quantum gates connected by qubit wires.
We choose to use the standard, general model of quantum information based on
density operators and quantum channels, as opposed to the restricted model of
pure state vectors and unitary operations, when discussing quantum circuits.
In this general model, each gate represents a quantum channel acting on a
constant number of qubits---including nonunitary gates, such as gates that
introduce fresh initialized qubits or gates that discard qubits.
Through this model, ordinary classical circuits, as well as classical circuits
that introduce randomness into computations, can be viewed as special cases of
quantum circuits.
One may also represent measurements directly as quantum gates or circuits.

It is well-known that this general model is equivalent to the purely unitary
model, as is explained in \cite{AharonovKN1998} and \cite{Watrous2009}, for
instance.
The main benefits of using the general model in the context of this paper
are that (i) it allows us to avoid having to constantly distinguish between
input qubits and ancillary qubits, or output qubits and garbage qubits, and
(ii)~it has the minor but nevertheless positive side effect of eliminating the
appearance of the irrational number $1/\sqrt{2}$ in many of the formulas that
will appear.

We choose a universal gate set from which all quantum circuits are assumed to
be composed.
The gates in this set include
\emph{Hadamard}, \emph{Toffoli}, and \emph{phase-shift} gates
(which induce the single-qubit unitary transformation determined by the
actions \mbox{$\ket{0} \mapsto \ket{0}$}
and $\ket{1} \mapsto i\ket{1}$),
as well as \emph{ancillary gates} and \emph{erasure gates}.
Ancillary gates take no input qubits and output a single qubit in the $\ket{0}$
state, while erasure gates take one input qubit and produce no output qubits,
and are described by the partial trace.
Any other choice for the unitary gates that is universal for quantum computing
could be taken instead, but the gate set just specified is both simple and
convenient.

The \emph{size} of a quantum circuit is defined to be the number of gates in
the circuit plus the total number of input and output qubits.
Thus, if a quantum circuit were to be represented in a standard way as a
directed acyclic graph, its size would simply be the number of vertices,
including a vertex for each input and output qubit, of the corresponding
graph.

A collection $\{Q_{x} : x \in \Sigma^{∗}\}$ of quantum circuits is said to be
\emph{polynomial-time generated} if there exists a polynomial-time
deterministic Turing machine that, on input $x \in \Sigma^{\ast}$, outputs an
encoding of the circuit $Q_{x}$.
When such a family is parameterized by tuples of strings, it is to be understood
that we are implicitly referring to one of the tuple-functions discussed
previously.
We will not have any need to discuss the specifics of the encoding scheme that
we use, but naturally it is assumed to be efficient, with the size of a circuit
and its encoding length being polynomially related.

The following lemma relates the complexity of computing circuit transition
amplitudes to $\class{GapP}$ functions.
The essential idea it expresses is due to Fortnow and Rogers
\cite{FortnowR1999}, who proved a variant of it for unitary computations by
quantum Turing machines.
While a result along the lines of the lemma that follows is suggested in
the survey paper \cite{Watrous2009}, that paper does not include a proof, and
so we include one below.

\begin{lemma}
  \label{lemma:circuit-GapP}
  Let $\{Q_x\,:\,x\in\Sigma^{\ast}\}$ be a polynomial-time generated family
  of quantum circuits, where each circuit $Q_x$ takes $n$ input qubits and
  outputs $k$ qubits, for polynomially bounded functions $n$ and $k$.
  There exists a polynomially bounded function $r$ and $\class{GapP}$ functions
  $f$ and $g$ such that
  \begin{equation}
    \label{eq:circuits entries GapP}
    \begin{aligned}
      \op{Re}\bigl(\bra{u} Q_x \bigl( \ket{z}\bra{w} \bigr) \ket{v}\bigr)
      & = 2^{-r} f_0(x,z,w,u,v),\\
      \op{Im}\bigl(\bra{u} Q_x \bigl( \ket{z}\bra{w} \bigr) \ket{v}\bigr)
      & = 2^{-r} f_1(x,z,w,u,v),
    \end{aligned}
  \end{equation}
  for all $x\in\Sigma^{\ast}$, $z,w\in\Sigma^n$, and $u,v\in\Sigma^k$.
\end{lemma}

\begin{proof}
  Consider first an arbitrary channel $\Phi$ that maps $n$-qubit density
  operators to $k$-qubit density operators.
  The action of $\Phi$ on density operators is linear, and can therefore be
  represented through matrix multiplication.
  One concrete way to do this is to use the so-called natural representation
  (also known as the linear representation) of quantum channels.

  A description of the natural representation of a quantum channel begins with
  the\linebreak[3]
  \emph{vectorization} mapping: assuming $M$ is a matrix whose rows and
  columns are indexed by strings of some length $m$, the corresponding vector
  $\op{vec}(M)$ is indexed by strings of length $2m$ according to the following
  definition:
  \begin{equation}
    \op{vec}(M) = \sum_{y,z\in\Sigma^m} \bra{y} M \ket{z} \, \ket{yz}.
  \end{equation}
  In words, the vectorization map reshapes a matrix into a vector by
  transposing the rows of the matrix into column vectors and stacking them on
  top of one another.

  With respect to the vectorization mapping, the action of the channel $\Phi$ is
  described by its \emph{natural representation} $K(\Phi)$, which is a linear
  mapping that acts as
  \begin{equation}
    K(\Phi) \op{vec}(\rho) = \op{vec}(\Phi(\rho))
  \end{equation}
  for every $n$-qubit density operator $\rho$.
  As a matrix, $K(\Phi)$ has columns indexed by strings of length $2n$ and rows
  indexed by strings of length $2k$.
  Its entries are described explicitly by the equation
  \begin{equation}
    \bra{uv} K(\Phi) \ket{zw} = \bra{u} \Phi(\ket{z}\bra{w}) \ket{v}
  \end{equation}
  holding for every $z,w\in\Sigma^n$ and $u,v\in\Sigma^k$.
  The equations \eqref{eq:circuits entries GapP} may therefore be equivalently
  written as
  \begin{equation}
    \label{eq:representations entries GapP}
    \begin{aligned}
      \op{Re}\bigl(\bra{uv} K(Q_x)\ket{zw}\bigr)
      & = 2^{-r} f_0(x,z,w,u,v),\\
      \op{Im}\bigl(\bra{uv} K(Q_x)\ket{zw}\bigr)
      & = 2^{-r} f_1(x,z,w,u,v).
    \end{aligned}
  \end{equation}
  It must be observed that the natural representation is multiplicative,
  in the sense that channel composition corresponds to matrix multiplication:
  $K(\Phi \Psi) = K(\Phi) K(\Psi)$ for all channels $\Phi$ and $\Psi$
  for which the composition $\Phi\Psi$ makes sense.
  It is also helpful to note that a channel $\Phi(\rho) = U \rho U^{\ast}$
  corresponding to a unitary operation has as its natural representation the
  operator
  \begin{equation}
    K(\Phi) = U \otimes \overline{U}.
  \end{equation}

  Now let us turn to the family $\{Q_x\,:\,x\in\Sigma^{\ast}\}$.
  Because this family is polynomial-time generated, there must exist a
  polynomially bounded function $r$ for which
  $\op{size}(Q_x) \leq r$ for all $x\in\Sigma^{\ast}$.
  We may therefore write
  \begin{equation}
    Q_x = Q_{x,r} \cdots Q_{x,1}
  \end{equation}
  for $Q_{x,1},\ldots,Q_{x,r}$ being either identity channels or channels that
  describe the action of a single gate of $Q_x$ tensored with the identity
  channel on all of the qubits besides the inputs of the corresponding gate
  that exist at the moment that the gate is applied.
  We also observe that the number of input qubits and output qubits of each
  $Q_{x,k}$ must be bounded by $r$.

  Given that
  \begin{equation}
    K(Q_x) = K(Q_{x,r}) \cdots K(Q_{x,1}),
  \end{equation}
  we are led to consider the natural representation of each channel $Q_{x,k}$.
  It will be convenient to identify each operator $K(Q_{x,k})$ with the matrix
  indexed by strings of length $2r$, as opposed to being indexed by strings
  whose lengths depend on the number of qubits in existence before and after
  $Q_{x,k}$ is applied, simply by padding $K(Q_{x,k})$ with rows and columns of
  zero entries.
  
  The natural representations of the individual gates in the universal gate set
  we have selected are as follows:
  \begin{mylist}{8mm}
  \item[1.] Hadamard gate:
    \begin{equation}
      \frac{1}{2}
      \begin{pmatrix}
        1 & 1 & 1 & 1\\
        1 & -1 & 1 & -1\\
        1 & 1 & -1 & -1\\
        1 & -1 & -1 & 1
      \end{pmatrix}
    \end{equation}
  \item[2.] Phase gate:
    \begin{equation}
      \begin{pmatrix}
        1 & 0 & 0 & 0\\
        0 & -i & 0 & 0\\
        0 & 0 & i & 0\\
        0 & 0 & 0 & 1
      \end{pmatrix}
    \end{equation}
  \item[3.]
    Toffoli gate:
    \begin{equation}
      \begin{pmatrix}
        1 & 0 & 0 & 0 & 0 & 0 & 0 & 0\\
        0 & 1 & 0 & 0 & 0 & 0 & 0 & 0\\
        0 & 0 & 1 & 0 & 0 & 0 & 0 & 0\\
        0 & 0 & 0 & 1 & 0 & 0 & 0 & 0\\
        0 & 0 & 0 & 0 & 1 & 0 & 0 & 0\\
        0 & 0 & 0 & 0 & 0 & 1 & 0 & 0\\
        0 & 0 & 0 & 0 & 0 & 0 & 0 & 1\\
        0 & 0 & 0 & 0 & 0 & 0 & 1 & 0
      \end{pmatrix}
      \otimes
      \begin{pmatrix}
        1 & 0 & 0 & 0 & 0 & 0 & 0 & 0\\
        0 & 1 & 0 & 0 & 0 & 0 & 0 & 0\\
        0 & 0 & 1 & 0 & 0 & 0 & 0 & 0\\
        0 & 0 & 0 & 1 & 0 & 0 & 0 & 0\\
        0 & 0 & 0 & 0 & 1 & 0 & 0 & 0\\
        0 & 0 & 0 & 0 & 0 & 1 & 0 & 0\\
        0 & 0 & 0 & 0 & 0 & 0 & 0 & 1\\
        0 & 0 & 0 & 0 & 0 & 0 & 1 & 0
      \end{pmatrix}
    \end{equation}
  \item[4.]
    Ancillary qubit gate:
    \begin{equation}
      \begin{pmatrix}
        1 \\ 0 \\ 0 \\ 0
      \end{pmatrix}
    \end{equation}
  \item[5.]
    Erasure gate:
    \begin{equation}
      \begin{pmatrix}
        1 & 0 & 0 & 1
      \end{pmatrix}
    \end{equation}
  \end{mylist}
  Based on these representations, it is straightforward to define
  $\class{GapP}$ functions (or, in fact, $\class{FP}$ functions) $g_0$ and
  $g_1$ such that
  \begin{equation}
    \begin{aligned}
      \op{Re}\bigl(\bra{uv} K(Q_{x,k})\ket{zw}\bigr)
      & = \frac{1}{2} \, g_0(x,z,w,u,v,y_k),\\[1mm]
      \op{Im}\bigl(\bra{uv} K(Q_{x,k})\ket{zw}\bigr)
      & = \frac{1}{2} \, g_1(x,z,w,u,v,y_k),
    \end{aligned}
  \end{equation}
  for all $x\in\Sigma^{\ast}$, $k\in\{1,\ldots,r\}$, and $u,v,z,w\in\Sigma^r$,
  where we write $y_1,\ldots,y_r$ to denote the elements of $\Sigma^r_1$ sorted
  in lexicographic order.
  It now follows through a straightforward application of
  Lemma~\ref{lemma:Gap.C matrix multiplication} there must exist
  $\class{GapP}$ functions $f_0$ and $f_1$ satisfying
  \eqref{eq:circuits entries GapP} and therefore 
  \eqref{eq:representations entries GapP}, for all $x\in\Sigma^{\ast}$,
  $z,w\in\Sigma^n$, and $u,v\in\Sigma^k$, as required.
\end{proof} 

\subsection*{A tail bound for operator-valued random variables}
\label{sec:tail-bound}

We will make use of the following tail bound on the minimum eigenvalue of
the average of a collection of operator-valued random variables.
This bound follows from a more general result due to Tropp.
In particular, the bound stated in the theorem below follows from Theorem 5.1
of \cite{Tropp2012} together with Pinsker's inequality, which relates the
relative entropy of two distributions to their total variation distance.

\begin{theorem}[Tropp]
  \label{theorem:eigenvalue-tail-bound}
  Let $d$ and $N$ be positive integers, let $\eta\in [0,1]$ and $\varepsilon>0$
  be real numbers, and let $X_1,\ldots,X_N$ be independent and identically
  distributed operator-valued random variables having the following
  properties:
  \begin{mylist}{8mm}
  \item[1.]
    Each $X_k$ takes $d\times d$ positive semidefinite operator values
    satisfying $X_k \leq \I$.
  \item[2.]
    The minimum eigenvalue of the expected operator $\op{E}(X_k)$ satisfies
    $\lambda_{\textup{min}}(\op{E}(X_k)) \geq \eta$.
  \end{mylist}
  It is the case that
  \begin{equation}
    \op{Pr}\biggl(\lambda_{\textup{min}}\biggl(\frac{X_1 + \cdots + X_N}{N}
    \biggr) < \eta - \varepsilon\biggr) \leq d \exp(-2N\varepsilon^2).
  \end{equation}
\end{theorem}

\section{Complexity classes for one-turn quantum refereed games}

In this section we define the complexity classes to be considered in this
paper: $\class{QRG}(1)$, $\class{CQRG}(1)$, and $\class{MQRG}(1)$.
The definitions of these classes all refer to the notion of a \emph{referee},
which (in this paper) is a polynomial-time generated family
\begin{equation}
  R = \{R_x\,:\,x\in\Sigma^{\ast}\}
\end{equation}
of quantum circuits having the following special form.
\begin{mylist}{8mm}
\item[1.]
  For each $x\in\Sigma^{\ast}$, the inputs to the circuit $R_x$ are
  grouped into two registers: an $n$ qubit register $\reg{A}$ and an $m$-qubit
  register $\reg{B}$, for polynomially bounded functions $n$ and $m$.
\item[2.]
  The output of each circuit $R_x$ is a single qubit, which is to be measured
  with respect to the standard basis immediately after the circuit is run.
\end{mylist}
Given that classical probabilistic states may be viewed as special cases of
quantum states (corresponding to diagonal density operators), this definition
of a referee can still be used in the situation in which either or both
of the registers $\reg{A}$ and $\reg{B}$ is constrained to initially store
a classical state.

We are interested in the situation that, for a given choice of an input string
$x\in\Sigma^{\ast}$, the input to the circuit $R_x$ is a product state of the
form $\rho\otimes\sigma$, where $\rho \in \Density(\A)$ is a state of the
register $\reg{A}$ and $\sigma\in\Density(\B)$ is a state of the register
$\reg{B}$.
The state $\rho \in \Density(\A)$ is to be viewed as representing the state
that Alice plays, while $\sigma \in \Density(\B)$ represents the state Bob
plays.
When the single output qubit of the circuit $R_x$ is measured with respect to
the standard basis, the outcome~1 is interpreted as ``Alice wins,'' while the
outcome 0 is interpreted as ``Bob wins.''

Now, consider the quantity defined as
\begin{equation}
  \omega(R_x) = \max_{\rho\in\Density(\A)} \min_{\sigma\in\Density(\B)}
  \bra{1} R_x (\rho \otimes \sigma) \ket{1}.
\end{equation}
Given that $\Density(\A)$ and $\Density(\B)$ are compact and convex sets, and
the value $\bra{1} R_x (\rho \otimes \sigma) \ket{1}$ is bilinear in
$\rho$ and $\sigma$, Sion's min-max theorem implies that changing the order of
the minimum and maximum does not change the value of the expression.
That is, this quantity may alternatively be written
\begin{equation}
  \omega(R_x) = \min_{\sigma\in\Density(\B)} \max_{\rho\in\Density(\A)} 
  \bra{1} R_x (\rho \otimes \sigma) \ket{1}.
\end{equation}
\pagebreak[3]

\noindent
This value represents the probability that Alice wins the game defined by the
circuit $R_x$, assuming both Alice and Bob play optimally.
With that definitions in hand, we may now define the complexity class
$\class{QRG}(1)$, which is short for \emph{one-turn quantum refereed games}.

\begin{definition}
  \label{definition:QRG(1)}
  A promise problem $A = (A_{\text{yes}},A_{\text{no}})$ is contained in the
  complexity class $\class{QRG}(1)_{\alpha,\beta}$ if there exists a referee
  $R = \{R_x\,:\,x\in\Sigma^{\ast}\}$ such that the following properties are
  satisfied:
  \begin{mylist}{8mm}
  \item[1.]
    For every string $x\in A_{\text{yes}}$, it is the case that
    $\omega(R_x) \geq \alpha$.
  \item[2.]
    For every string $x\in A_{\text{no}}$, it is the case that
    $\omega(R_x) \leq \beta$.
  \end{mylist}
  We also define $\class{QRG}(1) = \class{QRG}(1)_{2/3,1/3}$.
\end{definition}

\noindent
In this definition, $\alpha$ and $\beta$ may be constants, or they may be
functions of the length of the input $x$.
A short summary of known facts and observations concerning the
complexity class $\class{QRG}(1)$ follows.
\begin{mylist}{8mm}
\item[$\bullet$]
  $\class{QMA}\subseteq\class{QRG}(1)$.
  This is because the referee's measurement may simply ignore Bob's state
  $\sigma$ and treat Alice's state $\rho$ as a quantum proof in a QMA proof
  system.
  
\item[$\bullet$]
  $\class{QRG}(1)$ is closed under complementation:
  $\class{QRG}(1) = \text{co-}\class{QRG}(1)$.
  For a promise problem
  $(A_{\text{yes}},A_{\text{no}})\in\class{QRG}(1)$,
  one may obtain a one-turn quantum refereed game for
  $(A_{\text{no}},A_{\text{yes}})$ by simply exchanging the roles of Alice and
  Bob.
  
\item[$\bullet$]
  It is the case that $\class{QRG}(1)=\class{QRG}(1)_{\alpha,\beta}$ for a wide
  range of choices of $\alpha$ and $\beta$, similar to error bounds for
  $\class{BPP}$, $\class{BQP}$, and $\class{QMA}$.
  In particular, $\class{QRG}(1)=\class{QRG}(1)_{\alpha,\beta}$
  provided that $\alpha$ and $\beta$ are polynomial-time computable and
  satisfy
  \begin{equation}
    \alpha \leq 1 - 2^{-p},\quad
    \beta \geq 2^{-p},
    \quad \text{and} \quad
    \alpha - \beta \geq \frac{1}{p}
  \end{equation}
  for some choice of a strictly positive polynomially bounded function
  $p$.\footnote{Error reduction may be performed through parallel repetition
    followed by majority vote.
    An analysis of this method for $\class{QRG(1)}$ requires that one
    considers the possibility that the dishonest player (meaning the one that
    should not have a strategy that wins with high probability) entangles his
    or her state across the different repetitions, with the claimed bounds
    following from a similar analysis to parallel repetition followed by
    majority vote for $\class{QMA}$ \cite{KitaevSV2002}.
    We note that there is no ``in place'' error reduction method known for
    $\class{QRG(1)}$ that is analogous to the technique of
    \cite{MarriottW2005} for $\class{QMA}$.}
  
\item[$\bullet$]
  $\class{QRG}(1) \subseteq \class{PSPACE}$ \cite{JainW2009}.
\end{mylist}

The question that originally motivated the work reported in this paper is
whether the containment $\class{QRG}(1) \subseteq \class{PSPACE}$ can be
improved.
We do not succeed in improving this containment, but we are able to prove
stronger bounds on two interesting restricted variants of $\class{QRG}(1)$,
which we now define.

The first variant of $\class{QRG}(1)$ we define is one in which Alice's state
is restricted to be a classical state.
We will call this class $\class{CQRG}(1)$.

\begin{definition}
  \label{definition:CQRG(1)}
  A promise problem $A = (A_{\text{yes}},A_{\text{no}})$ is contained in the
  complexity class $\class{CQRG}(1)_{\alpha,\beta}$ if there exists a referee
  $R = \{R_x\,:\,x\in\Sigma^{\ast}\}$ such that the following properties are
  satisfied:
  \begin{mylist}{8mm}
  \item[1.]
    For every string $x\in\Sigma^{\ast}$, the circuit $R_x$ takes the form
    illustrated in Figure~\ref{fig:CQRG-referee}.
    That is, $R_{x}$ takes an $n$-qubit register $\reg{A}$ and an $m$-qubit
    register $\reg{B}$ as input, measures each qubit of $\reg{A}$ with respect
    to the standard basis, leaving it in a classical state, and then runs the
    circuit $Q_x$ on the pair $(\reg{A},\reg{B})$, producing a single output
    qubit.
  \item[2.]
    For every string $x\in A_{\text{yes}}$, it is the case that
    $\omega(R_x) \geq \alpha$.
  \item[3.]
    For every string $x\in A_{\text{no}}$, it is the case that
    $\omega(R_x) \leq \beta$.
  \end{mylist}
  We also define $\class{CQRG}(1) = \class{CQRG}(1)_{2/3,1/3}$.
\end{definition}

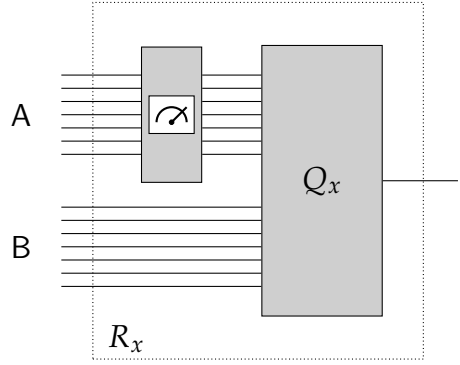
\begin{figure}[t]
  \begin{center}
    \begin{tikzpicture}[scale=1,
        circuit/.style={draw, minimum height=36mm, minimum width=16mm,
          fill = ChannelColor, text=ChannelTextColor},
        smallcircuit/.style={draw, minimum height=18mm, minimum width=8mm,
          fill = ChannelColor, text=ChannelTextColor},
        measure/.style={draw, minimum width=7mm, minimum height=7mm,
          fill = ChannelColor},
        >=latex]
      
      \node[draw, densely dotted, minimum width=125, minimum height=135]
      at (-0.85,0) {};
      
      \node at (-2.6,-2.1) {$R_x$};
      \node (R) at (0,0) [circuit] {$Q_x$};
      \node (Delta) at (-2,0) [smallcircuit,yshift=25] {};

      \node[draw, minimum width=6mm, minimum height=5mm, fill=ReadoutColor]
      (readout) at (Delta) {};
      \draw[thick] ($(Delta)+(0.2,-0.1)$) arc (0:180:2mm);
      \draw[thick] ($(Delta)+(0.2,0.1)$) -- ($(Delta)+(0,-0.1)$);
      \draw[fill] ($(Delta)+(0,-0.1)$) circle (0.3mm);
      
      \draw (R.east) -- ([xshift=10mm]R.east);

      \node[minimum width=30] (Alice) at (-4,0) [yshift=25] {$\reg{A}$};
      \node[minimum width=30] (Bob) at (-4,0) [yshift=-25] {$\reg{B}$};
      
      \foreach \y in {-15,-10,-5,0,5,10,15} {
        \draw ([yshift=\y]Alice.east) -- ([yshift=\y]Delta.west);
        \draw ([yshift=\y]Delta.east) -- ([yshift={\y+25}]R.west);
        \draw ([yshift=-\y]Bob.east) -- ([yshift={-25-\y}]R.west);
      }
    
    \end{tikzpicture}
  \end{center}

  \caption{A $\class{CQRG}(1)$ referee.
    The register $\reg{A}$ is initially measured (or, equivalently, dephased)
    with respect to the standard basis, causing a classical state to be input
    into $Q_x$, along with the register $\reg{B}$, which is unaffected by this
    standard basis measurement.}
  \label{fig:CQRG-referee}
\end{figure}

Formally speaking, the standard basis measurement suggested by
Definition~\ref{definition:CQRG(1)} can be implemented by independently
performing the completely dephasing channel on each qubit of $\reg{A}$.
This channel can be constructed using the universal gate set we have selected
using a Toffoli gate with suitably initialized inputs as follows:

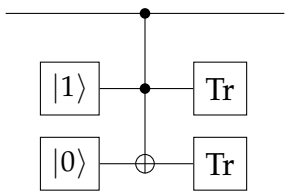
\begin{center}
\begin{tikzpicture}[scale=1]
  \node (In1) at (-2,1) {};
  \node (Out1) at (2,1) {};
  \node[minimum size=20pt] (A1) at (-1,0) [draw] {\small $\ket{1}$};
  \node[minimum size=20pt] (A2) at (-1,-1) [draw] {\small $\ket{0}$};
  \node[minimum size=20pt] (Tr1) at (1,0) [draw] {Tr};
  \node[minimum size=20pt] (Tr2) at (1,-1) [draw] {Tr};
  
  \node[circle, fill, minimum size = 4pt, inner sep=0mm]
  (Control1) at (0,1) {};
  
  \node[circle, fill, minimum size = 4pt, inner sep=0mm]
  (Control2) at (0,0) {};
  
  \node[circle, draw, minimum size = 7pt, inner sep=0mm]
  (Target1) at (0,-1) {};
  
  \draw (Control1.center) -- (Target1.south);
  \draw (A1.east) -- (Tr1.west);
  \draw (A2.east) -- (Tr2.west);
  \draw (In1) -- (Out1);

\end{tikzpicture}
\end{center}

\noindent
Here the square labeled $\ket{0}$ is an ancillary gate, the square labeled
$\ket{1}$ denotes an ancillary gate composed with a not-gate $X = H P P H$
(for $H$ and $P$ denoting Hadamard and phase-shift gates), and the square
labeled {Tr} denotes an erasure gate.

In effect, a referee $R$ that satisfies the first requirement of
Definition~\ref{definition:CQRG(1)} forces the state Alice plays to be a
classical state (i.e., a state represented by a diagonal density operator).
That is, for any density operator $\rho$ that Alice might choose to play,
the state of $\reg{A}$ that is input into $Q_x$ takes the form
\begin{equation}
  \label{eq:Alice-classical-state}
  \sum_{y\in\Sigma^n} p(y)\, \ket{y}\bra{y}
\end{equation}
for some probability vector $p$ over $n$-bit strings, and therefore
the state that is plugged into the top $n$ qubits of the circuit $Q_x$
represents a classical state.
Given that the standard basis measurement acts trivially on all diagonal
states, we observe that Alice may cause an arbitrary diagonal density operator
of the form \eqref{eq:Alice-classical-state} to be input into $Q_x$.
In short, the set of possible states that may be input into the top $n$ qubits
of the circuit $Q_x$ is precisely the set of diagonal $n$-qubit density
operators.

The second variant of $\class{QRG}(1)$ we define is one in which Alice and Bob
both send quantum states to the referee, but the referee first measures Alice's
state, obtaining a classical outcome, which is then measured together with
Bob's state (as illustrated in Figure~\ref{fig:MQRG-referee}).

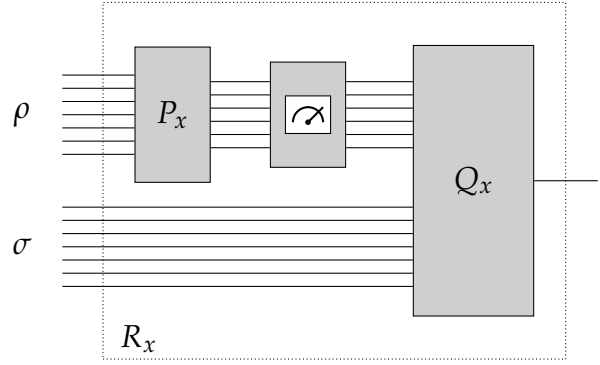
\begin{figure}[t]
  \begin{center}
    \begin{tikzpicture}[scale=1,
        circuit/.style={draw, minimum height=36mm, minimum width=16mm,
          fill = ChannelColor, text=ChannelTextColor},
        smallcircuit/.style={draw, minimum height=18mm, minimum width=10mm,
          fill = ChannelColor, text=ChannelTextColor},
        measure/.style={draw, minimum width=7mm, minimum height=7mm,
          fill = ChannelColor},
        >=latex]
      
      \node[draw, densely dotted, minimum width=175, minimum height=135]
      at (-1.85,0) {};

      \node at (-4.45,-2.1) {$R_x$};
      \node (R) at (0,0) [circuit] {$Q_x$};
      \node (Delta) at (-2.2,0) [smallcircuit,minimum height=14mm,yshift=25] {};

      \node[draw, minimum width=6mm, minimum height=5mm, fill=ReadoutColor]
      (readout) at (Delta) {};
      \draw[thick] ($(Delta)+(0.2,-0.1)$) arc (0:180:2mm);
      \draw[thick] ($(Delta)+(0.2,0.1)$) -- ($(Delta)+(0,-0.1)$);
      \draw[fill] ($(Delta)+(0,-0.1)$) circle (0.3mm);

      \node (M) at (1.8,0) {};
      \node (P) at (-4,0) [smallcircuit,yshift=25] {$P_x$};
      \draw (R.east) -- (M.west);
      
      \node[minimum width=30] (Alice) at (-6,0) [yshift=25] {$\rho$};
      \node[minimum width=30] (Bob) at (-6,0) [yshift=-25] {$\sigma$};
      
      \foreach \y in {-15,-10,-5,0,5,10,15} {
        \draw ([yshift=\y]Alice.east) -- ([yshift=\y]P.west);
        \draw ([yshift=-\y]Bob.east) -- ([yshift={-25-\y}]R.west);
      }
      
      \foreach \y in {-12.5,-7.5,-2.5,2.5,7.5,12.5} {
        \draw ([yshift=\y]P.east) -- ([yshift=\y]Delta.west);
        \draw ([yshift=\y]Delta.east) -- ([yshift={\y+25}]R.west);
      }
      
    \end{tikzpicture}
  \end{center}
  \caption{An $\class{MQRG}(1)$ referee.}
    \label{fig:MQRG-referee}
\end{figure}

\begin{definition}
  \label{definition:MQRG(1)}
  A promise problem $A = (A_{\text{yes}},A_{\text{no}})$ is contained in the
  complexity class $\class{MQRG}(1)_{\alpha,\beta}$ if there exists a referee
  $R = \{R_x\,:\,x\in\Sigma^{\ast}\}$
  such that the following properties are satisfied:
  \begin{mylist}{8mm}
  \item[1.]
    For every string $x\in\Sigma^{\ast}$, the circuit $R_x$ takes the form
    illustrated in Figure~\ref{fig:MQRG-referee}.
    That is, $R_{x}$ takes an $n$-qubit register $\reg{A}$ and an $m$-qubit
    register $\reg{B}$ as input, and first applies a quantum circuit $P_x$ to
    $\reg{A}$, yielding a $k$-qubit register $\reg{Y}$, for $k$ a polynomially
    bounded function.
    The register $\reg{Y}$ is then measured with respect to the standard basis,
    so that it then contains a classical state, and finally a quantum circuit
    $Q_x$ is applied to the pair $(\reg{Y}, \reg{B})$, yielding a single
    qubit.
  \item[2.]
    For every string $x\in A_{\text{yes}}$, it is the case that
    $\omega(R_x) \geq \alpha$.
  \item[3.]
    For every string $x\in A_{\text{no}}$, it is the case that
    $\omega(R_x) \leq \beta$.
  \end{mylist}
  We also define $\class{MQRG}(1) = \class{MQRG}(1)_{2/3,1/3}$.
\end{definition}

In essence, an $\class{MQRG}(1)$ referee measures Alice's qubits with respect
to a general, efficiently implementable measurement, which yields a $k$-bit
classical outcome, which is then plugged into $Q_{x}$ along with Bob's quantum
state.

It is of course immediate that
\begin{equation}
  \class{CQRG(1)} \subseteq \class{MQRG(1)} \subseteq \class{QRG(1)};
\end{equation}
a $\class{CQRG(1)}$ referee is a special case of an $\class{MQRG(1)}$
referee in which $P_{x}$ is the identity map on $n$ qubits, while an
$\class{MQRG(1)}$ referee is a special case of a $\class{QRG(1)}$ referee.
We also observe that both $\class{CQRG}(1)$ and $\class{MQRG}(1)$ are robust
with respect to error bounds in the same way as was described above for
$\class{QRG}(1)$.

\section{Upper-bound on CQRG(1)}

In this section, we prove that $\class{CQRG(1)}$ is contained in
$\exists \cdot \class{PP}$.
The proof represents a fairly direct application of the
Alth\"ofer--Lipton--Young \cite{Althofer1994,LiptonY1994} technique, although
(as was suggested above) the quantum setting places a new demand on this
technique that requires the use of a tail bound on sums of matrix-valued random
variables.
We will split the proof of this containment into two lemmas, followed by a
short proof of the main theorem---this is done primarily because the lemmas
will also be useful for proving
$\class{MQRG(1)} \subseteq \class{P} \cdot \class{PP}$ in the section following
this one.
Some readers may wish to skip to the statement and proof of
Theorem~\ref{theorem:CQRG(1) in NP.PP} below, as it explains the purpose of
these two lemmas within the context of that theorem.

The first lemma represents an implication of
Theorem~\ref{theorem:eigenvalue-tail-bound} due to Tropp to the setting at
hand.

\begin{lemma}
  \label{lemma:flat-distribution-quality}
  Let $k$ and $m$ be positive integers, let $p\in\P(\Sigma^k)$ be a
  probability distribution on $k$-bit strings, let $S_y$ be a
  $2^m\times 2^m$ positive semidefinite operator satisfying $0\leq S_y \leq \I$
  for each $y\in\Sigma^k$, and let
  $N \geq 72(m+2)$.
  For strings $y_1,\ldots,y_N\in\Sigma^k$ sampled independently from the
  distribution~$p$, it is the case that
  \begin{equation}
    \op{Pr}\Biggl(
    \lambda_{\text{min}}\Biggl(\frac{S_{y_1} + \cdots + S_{y_N}}{N}\Biggr)
    < \lambda_{\text{min}}\Biggl(\sum_{y\in\Sigma^k}p(y)S_y\Biggr)
    - \frac{1}{12}\Biggr) < \frac{1}{3}.
  \end{equation}
\end{lemma}

\begin{proof}
    Define $X_1,\ldots,X_N$ to be independent and identically distributed
  operator-valued random variables, each taking the (operator) value $S_y$
  with probability $p(y)$, for every $y\in\Sigma^k$.
  The expected value of each of these random variables is therefore given by
  \begin{equation}
    P = \sum_{y\in\Sigma^k}p(y)S_y.
  \end{equation}
  By taking $\eta = \lambda_{\text{min}}(P)$ and $\varepsilon = 1/12$ in
  Theorem~\ref{theorem:eigenvalue-tail-bound}, we find that
  \begin{equation}
    \op{Pr}\biggl(
    \lambda_{\textup{min}}\biggl(
    \frac{X_1 + \cdots + X_N}{N}
    \biggr) < \lambda_{\text{min}}(P) - \frac{1}{12}\biggr)
    \leq 2^m \exp\biggl(-\frac{N}{72}\biggr) < \frac{1}{3},
  \end{equation}
  which is equivalent to the bound stated in the lemma.
\end{proof}

The second lemma uses counting complexity to relate the minimum eigenvalue of
measurement operators defined by quantum circuits to PP languages.
We note that the technique of weakly estimating the largest eigenvalue
of a measurement operator using the trace of a power of that operator, through
the relations \eqref{eq:max-eigenvalue versus trace} appearing in the proof
below, is the essential idea behind the unpublished proof of the containment
$\class{QMA}\subseteq\class{PP}$ claimed in \cite{KitaevW00}.

\begin{lemma}
  \label{lemma:PP-eigenvalue-estimator}
  Let $\{Q_x\,:\,x\in\Sigma^{\ast}\}$ be a polynomial-time generated family of
  quantum circuits, where each circuit $Q_x$ takes as input a $k$-qubit
  register $\reg{Y}$ and an $m$-qubit register $\reg{B}$, for polynomially
  bounded functions $k$ and $m$, and outputs a single qubit.
  For each $x\in\Sigma^{\ast}$ and $y\in\Sigma^k$, define an operator
  \begin{equation}
    S_{x,y} = \bigl(\bra{y} \otimes \I_{\B}\bigr)
    Q_x^{\ast} ( \ket{1}\bra{1} )
    \bigl(\ket{y} \otimes \I_{\B}\bigr).
  \end{equation}
  For every polynomially bounded function~$N$, there exists
  a language $B\in\class{PP}$ for which the following implications are true
  for all $x\in\Sigma^{\ast}$ and $y_1,\ldots,y_N\in\Sigma^k$:
  \begin{align}    
    \lambda_{\text{min}}\biggl(\frac{S_{x,y_1} + \cdots + S_{x,y_N}}{N}\biggr)
    \geq \frac{2}{3}
    & \quad \Rightarrow \quad (x,y_1\cdots y_N) \in
    B,\label{eq:eigenvalue-implication-1}\\[2mm]
    \lambda_{\text{min}}\biggl(\frac{S_{x,y_1} + \cdots + S_{x,y_N}}{N}\biggr)
    \leq \frac{1}{3}
    & \quad \Rightarrow \quad (x,y_1\cdots y_N) \not\in B.
    \label{eq:eigenvalue-implication-2}
  \end{align}
\end{lemma}

\begin{proof}
  The essence of the proof is that if $P$ is an operator whose entries have
  real and imaginary parts proportional to $\class{GapP}$ functions, and $r$ is
  a polynomially bounded function, then there exists a $\class{GapP}$ function
  that is proportional to the real part of $\tr(P^r)$.
  When $P$ is a $2^m \times 2^m$ positive semidefinite operator, this allows one
  to choose $r$ to be sufficiently large, but still polynomially bounded,
  so that a $\class{GapP}$ function is obtained that takes positive or
  negative values in accordance with the required implications
  \eqref{eq:eigenvalue-implication-1} and \eqref{eq:eigenvalue-implication-2},
  through the use of the following bounds relating the largest eigenvalue and
  the trace of any such~$P$:
  \begin{equation}
    \label{eq:max-eigenvalue versus trace}
    \lambda_{\text{max}}(P)^r
    = \lambda_{\text{max}}(P^r)
    \leq \tr(P^r) \leq 2^m \lambda_{\text{max}}(P^r) =
    2^m \lambda_{\text{max}}(P)^r.
  \end{equation}
  In the case at hand, it will suffice to take $r = 2m$.

  In greater detail, let us begin by defining
  \begin{equation}
    T_{x,y} = \bigl(\bra{y} \otimes \I_{\B}\bigr)
    Q_x^{\ast} ( \ket{0}\bra{0} )
    \bigl(\ket{y} \otimes \I_{\B}\bigr)
  \end{equation}
  for each $x\in\Sigma^{\ast}$ and $y\in\Sigma^k$.
  Observe that $S_{x,y}$ and $T_{x,y}$ are positive semidefinite operators
  satisfying $S_{x,y} + T_{x,y} = \I_{\B}$, so that the implication in the
  statement of the lemma may alternatively be written as
  \begin{align}    
    \lambda_{\text{max}}\biggl(\frac{T_{x,y_1} + \cdots + T_{x,y_N}}{N}\biggr)
    \leq \frac{1}{3}
    & \quad \Rightarrow \quad (x,y_1\cdots y_N) \in B,\\[2mm]
    \lambda_{\text{max}}\biggl(\frac{T_{x,y_1} + \cdots + T_{x,y_N}}{N}\biggr)
    \geq \frac{2}{3}
    & \quad \Rightarrow \quad (x,y_1\cdots y_N) \not\in B.
  \end{align}
  Thus, if the operator
  \begin{equation}
    P_{x,y_1\cdots y_N} = \frac{T_{x,y_1} + \cdots + T_{x,y_N}}{N}
  \end{equation}
  satisfies $\lambda_{\text{max}}(P_{x,y_1\cdots y_N})\leq 1/3$, then
  \begin{equation}
    \tr(P_{x,y_1\cdots y_N}^{2m}) \leq \frac{2^m}{3^{2m}} < \frac{1}{3^m}
  \end{equation}
  while if $\lambda_{\text{max}}(P_{x,y_1\cdots y_N})\geq 2/3$, then
  \begin{equation}
    \tr(P_{x,y_1\cdots y_N}^{2m}) \geq \Bigl(\frac{2}{3}\Bigr)^{2m}
    > \frac{1}{3^m}.
  \end{equation}
  
  By Lemma~\ref{lemma:circuit-GapP} there exists a polynomially bounded
  function $r$ along with $\class{GapP}$ functions $f$ and $g$ satisfying
  \begin{equation}
    \begin{aligned}
      \op{Re}(\bra{z} T_{x,y} \ket{w})
      & = \op{Re}\bigl(\bra{0} Q_x \bigl( \ket{yz}\bra{yw} \bigr) \ket{0}\bigr)
      = 2^{-r} f(x,y,z,w),\\
      \op{Im}(\bra{z} T_{x,y} \ket{w})
      & = -\op{Im}\bigl(\bra{0} Q_x \bigl( \ket{yz}\bra{yw} \bigr) \ket{0}\bigr)
      = 2^{-r} g(x,y,z,w),
    \end{aligned}
  \end{equation}
  for all $x\in\Sigma^{\ast}$, $y\in\Sigma^k$, and $z,w\in\Sigma^m$.
  Define functions $F$ and $G$ as follows:
  \begin{equation}
    \begin{aligned}
      F(x,y_1\cdots y_N,z,w) & = f(x,y_1,z,w) + \cdots + f(x,y_N,z,w),\\
      G(x,y_1\cdots y_N,z,w) & = g(x,y_1,z,w) + \cdots + g(x,y_N,z,w),
    \end{aligned}
  \end{equation}
  for all $x\in\Sigma^{\ast}$, $y_1,\ldots,y_N\in\Sigma^k$, and
  $z,w\in\Sigma^m$.  
  It is the case that $F$ and $G$ are $\class{GapP}$ functions satisfying
  \begin{equation}
    \begin{aligned}
      F(x,y_1\cdots y_N,z,w)
      & = 2^{r}\cdot N \cdot \op{Re}(\bra{z}P_{x,y_1\cdots y_N}\ket{w}),\\
      G(x,y_1\cdots y_N,z,w)
      & = 2^{r}\cdot N \cdot \op{Im}(\bra{z}P_{x,y_1\cdots y_N}\ket{w}).
    \end{aligned}
  \end{equation}

  Through an application of
  Lemmas~\ref{lemma:Gap.C closed under exponential sums} and
  \ref{lemma:Gap.C matrix multiplication}, we conclude that there 
  must exist a $\class{GapP}$ function $H$ satisfying
  \begin{equation}
    H(x,y_1\cdots y_N)
    = 2^{2 r m} \cdot N^{2m} \cdot \tr\bigl(P_{x,y_1\cdots y_N}^{2m}\bigr).
  \end{equation}
  The $\class{GapP}$ function
  \begin{equation}
    K(x,y_1\cdots y_N) = 
    2^{2rm} \cdot N^{2m} - 3^m \cdot H(x,y_1\cdots y_N)
  \end{equation}
  therefore takes positive values if
  $\lambda_{\text{max}}(P_{x,y_1\cdots y_N})\leq 1/3$, and takes negative
  values if $\lambda_{\text{max}}(P_{x,y_1\cdots y_N})\geq 2/3$, implying the
  existence of a $\class{PP}$ language $B$ as claimed.  
\end{proof}

\pagebreak[3]

\begin{theorem}
  \label{theorem:CQRG(1) in NP.PP}
  $\class{CQRG}(1) \subseteq \exists\cdot\class{PP}$.
\end{theorem}

\begin{proof}
  Let $A = (A_{\text{yes}},A_{\text{no}})$ be any promise problem contained in
  $\class{CQRG}(1)$, let a referee be fixed that establishes the inclusion
  $A \in \class{CQRG}(1)_{3/4,1/4}$, and let $\{Q_x\,:\,x\in\Sigma^{\ast}\}$ be
  the collection of circuits that describes this referee, in accordance with
  Definition~\ref{definition:CQRG(1)}.

  Let $x\in A_{\text{yes}} \cup A_{\text{no}}$ be any input string.
  Consider first the situation that Alice plays deterministically, sending a
  string $y\in\Sigma^n$ to the referee, so that $\rho = \ket{y}\bra{y}$.
  Having selected a state $\rho$ representing Alice's play, we are effectively
  left with a binary-valued measurement being performed on the state sent to
  the referee by Bob.
  We observe that, for any choice of a state $\sigma\in\Density(\cal{B})$
  representing Bob's play, the probabilities that the referee's measurement
  generates the outcomes 0 and~1 are given by
  \begin{equation}
    \bra{0} Q_x ( \ket{y}\bra{y} \otimes \sigma) \ket{0}
    \quad\text{and}\quad
    \bra{1} Q_x ( \ket{y}\bra{y} \otimes \sigma) \ket{1},
  \end{equation}
  respectively.
  By defining an operator $S_{x,y} \in \Pos(\B)$ as
  \begin{equation}
    S_{x,y} = \bigl(\bra{y} \otimes \I_{\B}\bigr)
    Q_x^{\ast} ( \ket{1}\bra{1} ) \bigl(\ket{y} \otimes \I_{\B}\bigr),
  \end{equation}
  we therefore obtain the measurement operator corresponding to the 1 outcome
  of this measurement, as
  \begin{equation}
    \label{definitionofsxy}
    \tr\bigl( S_{x,y}\tsp \sigma\bigr)
    = \bra{1} Q_x ( \ket{y}\bra{y} \otimes \sigma) \ket{1}
  \end{equation}
  and
  \begin{equation}
    \tr\bigl( (\I_{\B} - S_{x,y})\tsp \sigma\bigr) =
    \bra{0} Q_x ( \ket{y}\bra{y} \otimes \sigma) \ket{0}
  \end{equation}
  for all $\sigma \in \Density(\B)$.
  
  Now, as Bob aims to minimize the probability for outcome 1 to appear, the
  relevant property of the operator $S_{x,y}$ is its \emph{minimum eigenvalue}
  $\lambda_{\text{min}}(S_{x,y})$.
  A large minimum eigenvalue means that Alice has managed to force the outcome
  1 to appear, regardless of what state Bob plays, whereas a small minimum
  eigenvalue means that Bob has at least one choice of a state that causes the
  outcome 1 to appear with small probability.
  Stated in more precise terms, Bob's optimal strategy in the case that Alice
  plays $\rho = \ket{y}\bra{y}$ is to play any state $\sigma\in\Density(\B)$
  whose image is contained in the eigenspace of $S_{x,y}$ corresponding to the
  minimum eigenvalue $\lambda_{\text{min}}(S_{x,y})$, which leads to a win for
  Alice with probability equal to this minimum eigenvalue and a win for Bob with
  probability $1 - \lambda_{\text{min}}(S_{x,y})$.

  In general, Alice will not play deterministically, but will instead play
  a distribution of strings $p\in\P(\Sigma^n)$.
  In this case, the resulting measurement operator on Bob's space becomes
  \begin{equation}
    \label{eq:expected-operator-p}
    \sum_{y\in\Sigma^n} p(y) S_{x,y}.
  \end{equation}
  That is to say, the probability that Alice wins when she plays a distribution
  $p\in\P(\Sigma^n)$, and Bob plays optimally against this distribution,
  is given by the expression
  \begin{equation}
    \label{eq:min-eigenvalue-distribution}
    \lambda_{\text{min}}\Biggl(\sum_{y\in\Sigma^n} p(y) S_{x,y}\Biggr).
  \end{equation}
  Determining whether $x$ is a yes-instance or a no-instance of $A$ is
  therefore equivalent to discriminating between the case that there exists a
  distribution $p\in\P(\Sigma^n)$ for which the minimum eigenvalue
  \eqref{eq:min-eigenvalue-distribution} is at least $3/4$ and the case in
  which this minimum eigenvalue is at most $1/4$ for all choices of
  $p\in\P(\Sigma^n)$.

  The goal of the proof is to show that this decision problem is contained
  in $\exists\cdot\class{PP}$.
  The $\exists$ operator will represent the existence or non-existence of a
  distribution $p\in\P(\Sigma^n)$ for which the minimum eigenvalue
  \eqref{eq:min-eigenvalue-distribution} is large, while a PP predicate will
  allow for an estimation of this minimum eigenvalue itself.
  A challenge that must be overcome in making this approach work is that using
  the $\exists$ operator in this way requires Alice's strategy to have a
  polynomial-length representation.
  However, given that a distribution $p\in\P(\Sigma^n)$ may have support that
  is exponentially large in $n$, an explicit description of $p$ will generally
  have exponential size, assuming that the individual probabilities $p(y)$ are
  represented with a polynomial number of bits of precision.

  This obstacle may be overcome using the Alth\"ofer--Lipton--Young
  \cite{Althofer1994,LiptonY1994} technique mentioned in the introduction:
  in place of a distribution $p\in\P(\Sigma^n)$, we consider an $N$-tuple of
  strings $(y_1,\ldots,y_N)$, representing $N$ possible deterministic plays for
  Alice, for $N = N(\abs{x})$ being a suitable polynomially bounded function of
  the input length.
  This $N$-tuple will represent the distribution $q\in\P(\Sigma^n)$ obtained by
  selecting $j\in\{1,\ldots,N\}$ uniformly at random and then outputting the
  string $y_j$.
  That is, the distribution $q\in\P(\Sigma^n)$ represented by the $N$-tuple
  $(y_1,\ldots,y_N)$ is given by
  \begin{equation}
    \label{eq:flat-poly-distribution}
    q(y) = \frac{\bigabs{\{j\in\{1,\ldots,N\}\,:\,y = y_j\}}}{N}
  \end{equation}
  for each $y\in\Sigma^n$.
  Naturally, most choices of a distribution $p\in\P(\Sigma^n)$ are far away
  from any such distribution $q$.
  Nevertheless, the existence of a distribution $p\in\P(\Sigma^n)$ for which
  the minimum eigenvalue \eqref{eq:min-eigenvalue-distribution} is large does
  in fact imply the existence of an $N$-tuple $(y_1,\ldots,y_N)$ for which the
  distribution $q\in\P(\Sigma^n)$ defined by \eqref{eq:flat-poly-distribution}
  is still a good play for Alice, meaning that the minimum eigenvalue
  \begin{equation}
    \label{eq:min-eigenvalue-flat}
    \lambda_{\text{min}}\biggl(\frac{S_{x,y_1} + \cdots + S_{x,y_N}}{N}\biggr)
  \end{equation}
  is also large, provided $N$ is sufficiently large.
  This is precisely the content of Lemma~\ref{lemma:flat-distribution-quality}.

  In particular, by choosing $N = 72(m+2)$, where $m$ is the number of qubits
  of $\reg{B}$, we find that if the minimum eigenvalue
  \eqref{eq:min-eigenvalue-distribution} is at least 3/4, then with probability
  at least 2/3 (over the random choices of $y_1,\ldots,y_N$) the minimum
  eigenvalue \eqref{eq:min-eigenvalue-distribution} is at least 2/3.
  Of course, this implies the existence of an $N$-tuple
  $(y_1,\ldots,y_N)$ for which the minimum eigenvalue
  \eqref{eq:min-eigenvalue-distribution} is at least 2/3.

  Naturally, if $x\in A_{\text{no}}$, then the minimum eigenvalue
  \eqref{eq:min-eigenvalue-distribution} is at most $1/4$ for all choices of
  $p\in\P(\Sigma^n)$, and therefore it must be that
  \begin{equation}
    \label{eq:no-eigenvalue-bound}
    \lambda_{\text{min}}\biggl(\frac{S_{x,y_1} + \cdots + S_{x,y_N}}{N}\biggr)
    \leq \frac{1}{4} < \frac{1}{3}
  \end{equation}
  for all $N$-tuples $(y_1,\ldots,y_N)$.
  This is because the distribution $q$ defined by
  \eqref{eq:flat-poly-distribution} is simply one example of a distribution
  in $\P(\Sigma^n)$.

  The purpose of Lemma~\ref{lemma:PP-eigenvalue-estimator}
  is now evident, for it states that there exists a language
  $B\in\class{PP}$ such that if the minimum eigenvalue
  \eqref{eq:min-eigenvalue-flat} is at least $2/3$, then
  $(x,y_1\cdots y_N)\in B$, while if this minimum eigenvalue is at most
  $1/3$, then $(x,y_1\cdots y_N)\not\in B$.
  Consequently, if $x\in A_{\text{yes}}$, then there exists a string
  $y_1\cdots y_N\in\Sigma^{nN}$ such that $(x,y_1\cdots y_N)\in B$, while
  if $x\in A_{\text{no}}$, then for every string
  $y_1\cdots y_N\in\Sigma^{nN}$ it is the case that
  $(x,y_1\cdots y_N)\not\in B$.
  It has therefore been proved that $A\in \exists\cdot\class{PP}$ as required.
\end{proof}
 
\section{Upper-bound on MQRG(1)}

We now turn to the complexity class $\class{MQRG}(1)$, and prove the
containment $\class{MQRG}(1) \subseteq \class{P}\cdot\class{PP}$.
In order to do this, we will first introduce a $\class{QMA}$-operator that,
in some sense, functions in a way that is similar to the $\exists$ and
$\class{P}$ operators previously discussed.

\begin{figure}[t]
  \begin{center}
    \begin{tikzpicture}[scale=1,
        circuit/.style={draw, minimum height=36mm, minimum width=16mm,
          fill = ChannelColor, text=ChannelTextColor},
        smallcircuit/.style={draw, minimum height=18mm, minimum width=10mm,
          fill = ChannelColor, text=ChannelTextColor},
        measure/.style={minimum width=7mm, minimum height=7mm},
        >=latex]
      
      \node (P) at (0,0) [smallcircuit] {$\chi_B$};
      \node (Delta) at (-1.5,0) [smallcircuit,minimum height=14mm] {};
            
      \node (M) at (Delta) [measure] {};
      \node[draw, minimum width=6mm, minimum height=5mm, fill=ReadoutColor]
      (readout) at (M) {};
      \draw[thick] ($(M)+(0.2,-0.1)$) arc (0:180:2mm);
      \draw[thick] ($(M)+(0.2,0.1)$) -- ($(M)+(0,-0.1)$);
      \draw[fill] ($(M)+(0,-0.1)$) circle (0.3mm);
      
      \node (Q) at (-3,0) [smallcircuit] {$P_x$};
   
      \node[minimum width=30] (Alice) at (-4.5,0) {$\rho$};
      
      \foreach \y in {-15,-10,-5,0,5,10,15} {
        \draw ([yshift=\y]Alice.east) -- ([yshift=\y]Q.west);
      }
      
      \foreach \y in {-12.5,-7.5,-2.5,2.5,7.5,12.5} {
        \draw ([yshift=\y]Q.east) -- ([yshift=\y]Delta.west);
        \draw ([yshift=\y]Delta.east) -- ([yshift=\y]P.west);
      }

      \draw (P.east) -- ([xshift=4mm]P.east);
      
    \end{tikzpicture}
  \end{center}
  \caption{Definition~\ref{definition:QMA.C} is concerned with the probability
    that the output of a circuit $P_x$, measured with respect to the standard
    basis, is contained in the language $B$, assuming the input is $\rho$.}
  \label{fig:QMA with post-measurement predicate}
\end{figure}
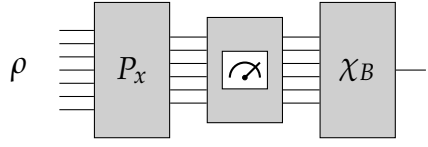

\begin{definition}
  \label{definition:QMA.C}
  For a given complexity class $\C$, the complexity class $\class{QMA}\cdot\C$
  contains all promise problems $A = (A_{\text{yes}},A_{\text{no}})$ for which
  there exists a polynomial-time generated family of quantum circuits
  $\{P_x\,:\,x\in\Sigma^{\ast}\}$, where each $P_x$ takes $n = n(\abs{x})$
  input qubits and outputs $k = k(\abs{x})$ qubits, along with a language
  $B\in\C$, such that the following implications hold.
  \begin{mylist}{8mm}
  \item[1.]
    If $x\in A_{\text{yes}}$, then there exists a density operator
    $\rho$ on $n$ qubits for which
    \begin{equation}
      \op{Pr}\bigl( P_x(\rho) \in B\bigr) \geq \frac{2}{3}.
    \end{equation}
  \item[2.]
    If $x\in A_{\text{no}}$, then for every density operator
    $\rho$ on $n$ qubits,
    \begin{equation}
      \op{Pr}\bigl( P_x(\rho) \in B\bigr) \leq \frac{1}{3}.
    \end{equation}
  \end{mylist}
  Here, the notation $P_x(\rho)\in B$ refers to the event that $P_x$ is
  applied to the state $\rho$, the output qubits are measured with respect to
  the standard basis, and the resulting string is contained in the language
  $B$.
  Figure~\ref{fig:QMA with post-measurement predicate} illustrates the
  associated process, with $\chi_B$ being the characteristic function of $B$
  on inputs of length $k$.
\end{definition}

\begin{theorem}
  \label{theorem:QMA.C-in-PP.C}
  If $\C$ is nontrivial complexity class of languages that is closed under
  joins and truth-table reductions, then
  $\class{QMA}\cdot\C \subseteq \class{P}\cdot\C$.
\end{theorem}

\begin{proof}
  Let $A = (A_{\text{yes}},A_{\text{no}})\in\class{QMA}\cdot\C$, and let
  $\{P_x\,:\,x\in\Sigma^{\ast}\}$ be a
  polynomial-time generated family of quantum circuits that, together with a
  language $B\in\C$, establishes this inclusion according to
  Definition~\ref{definition:QMA.C}.

  By Lemma~\ref{lemma:circuit-GapP} there exists a polynomially bounded
  function $r$ and $\class{GapP}$ functions $f_0$ and $f_1$ such that
  \begin{equation}
    \begin{aligned}
      \op{Re}\bigl(\bra{u} P_x \bigl( \ket{z}\bra{w} \bigr) \ket{v}\bigr)
      & = 2^{-r} f(x,z,w,u,v),\\
      \op{Im}\bigl(\bra{u} P_x \bigl( \ket{z}\bra{w} \bigr) \ket{v}\bigr)
      & = 2^{-r} g(x,z,w,u,v),
    \end{aligned}
  \end{equation}
  for all $x\in\Sigma^{\ast}$, $z,w\in\Sigma^n$, and $u,v\in\Sigma^k$.
  Define
  \begin{equation}
    \begin{aligned}
      g_0(x,z,w,u) & =
      \begin{cases}
        f_0(x,z,w,u,u) & \text{if $u\in B$}\\
        0 & \text{if $u\not\in B$},
      \end{cases}\\[2mm]
      g_1(x,z,w,u) & =
      \begin{cases}
        f_1(x,z,w,u,u) & \text{if $u\in B$}\\
        0 & \text{if $u\not\in B$},
      \end{cases}
    \end{aligned}
  \end{equation}
  for all $x\in\Sigma^{\ast}$, $z,w\in\Sigma^n$, and $u\in\Sigma^k$.
  By the nontriviality and closure of $\C$ under Karp reductions
  (the full power of closure under joins and truth-table reductions is not
  required for this step), it is the case that $g_0,g_1\in\class{Gap}\cdot\C$.

  Next, define
  \begin{equation}
    \begin{aligned}
      F_0(x,z,w) & = \sum_{u\in\Sigma^k} g_0(x,z,w,u),\\
      F_1(x,z,w) & = -\sum_{u\in\Sigma^k} g_1(x,z,w,u),
    \end{aligned}
  \end{equation}
  for all $x\in\Sigma^{\ast}$ and $z,w\in\Sigma^n$.
  By Lemma~\ref{lemma:Gap.C closed under exponential sums} we have that
  $F_0,F_1\in\class{Gap}\cdot\C$.
  We observe that
  \begin{equation}
    \begin{aligned}
      \op{Re} \bigl( \bra{z} R_x \ket{w} \bigr) & = 2^{-r} F_0(x,z,w),\\
      \op{Im} \bigl( \bra{z} R_x \ket{w} \bigr) & = 2^{-r} F_1(x,z,w)
    \end{aligned}
  \end{equation}
  for all $x\in\Sigma^{\ast}$ and $z,w\in\Sigma^n$, where
  \begin{equation}
    R_x = \sum_{u\in\Sigma^k\cap B} P_x^{\ast} \bigl( \ket{u}\bra{u} \bigr).
  \end{equation}

  Now let us consider the cases $x\in A_{\text{yes}}$ and
  $x\in A_{\text{no}}$.
  If $x\in A_{\text{yes}}$ then $\lambda_{\text{max}}(R_x) \geq 2/3$,
  while if $x\in A_{\text{no}}$ then $\lambda_{\text{max}}(R_x) \leq 1/3$.
  Observing that $R_x$ is a positive semidefinite operator on a $2^n$
  dimensional space, we have that
  \begin{equation}
    \lambda_{\text{max}}(R_x)^{n+1}
    = \lambda_{\text{max}}(R_x^{n+1})
    \leq \tr(R_x^{n+1})
    \leq 2^n \lambda_{\text{max}}(R_x^{n+1})
    = 2^n \lambda_{\text{max}}(R_x)^{n+1},
  \end{equation}
  similar to equation \eqref{eq:max-eigenvalue versus trace} in the proof of
  Lemma~\ref{lemma:PP-eigenvalue-estimator}.
  By Lemma~\ref{lemma:Gap.C matrix multiplication} it follows that there
  exists a $\class{Gap}\cdot\C$ function $G$ possessing the following
  properties.
  \begin{mylist}{8mm}
  \item[1.]
    If $x\in A_{\text{yes}}$ then
    \begin{equation}
      G(x) = 2^{(n+1)r} \op{tr}(R_x^{n+1}) \geq \frac{2^{(n+1)r+n+1}}{3^{n+1}}
    \end{equation}
  \item[2.]
    If $x\in A_{\text{no}}$ then
    \begin{equation}
      G(x) = 2^{(n+1)r} \op{tr}(R_x^{n+1}) \leq \frac{2^{(n+1)r + n}}{3^{n+1}}.
    \end{equation}
  \end{mylist}
  The $\class{Gap}\cdot\C$ function
  \begin{equation}
    H(x) = 3^{n+1} G(x) - 2^{(n+1)r + n}
  \end{equation}
  therefore satisfies $H(x) > 0$ when $x\in A_{\text{yes}}$ and
  $H(x) \leq 0$ when $x\in A_{\text{no}}$.
  By Proposition~\ref{prop:Gap.C versus P.C} it follows that
  $A\in\class{P}\cdot\C$.  
\end{proof}

Next, we prove that $\class{MQRG}(1)$ is contained in
$\class{QMA}\cdot\class{PP}$.
Combining this fact with the previous theorem will establish the main result
as an immediate corollary.

\begin{theorem}
  \label{theorem:MQRG(1) in QMA.PP}
  $\class{MQRG}(1) \subseteq \class{QMA}\cdot\class{PP}$.
\end{theorem}

\begin{proof}
  Consider any promise problem $A = (A_{\text{yes}}, A_{\text{no}})$ in
  $\class{MQRG(1)}$, and fix a referee that establishes the
  inclusion $A \in \class{MQRG(1)}_{3/4,1/4}$.
  Let $\{P_x : x\in\Sigma^{\ast}\}$ and $\{Q_{x} : x\in\Sigma^{\ast}\}$
  be a collection of circuits that describe this referee, in accordance with
  Definition~\ref{definition:MQRG(1)}.
  As in the proof of Theorem~\ref{theorem:CQRG(1) in NP.PP}, define an
  operator
  \begin{equation}
    S_{x,y} = \bigl(\bra{y} \otimes \I_{\B}\bigr)
    Q_x^{\ast} ( \ket{1}\bra{1} )
    \bigl(\ket{y} \otimes \I_{\B}\bigr)
  \end{equation}
  for each $x\in\Sigma^{\ast}$ and $y\in\Sigma^k$.
  If $x\in A_{\text{yes}}$, there must exists a state $\rho\in\Density(\A)$
  such that
  \begin{equation}
    \label{eq:MQRG-Alice-wins}
    \lambda_{\text{min}}
    \Biggl(
    \sum_{y\in\Sigma^k}
    \bra{y}P_x(\rho)\ket{y} S_{x,y}\Biggr)\geq \frac{3}{4},
  \end{equation}
  while if $x\in A_{\text{no}}$, it is the case that
  \begin{equation}
    \label{eq:MQRG-Bob-wins}
    \lambda_{\text{min}}
    \Biggl(
    \sum_{y\in\Sigma^k}
    \bra{y}P_x(\rho)\ket{y} S_{x,y}\Biggr) \leq \frac{1}{4}
  \end{equation}
  for every $\rho\in\Density(\A)$.

  Now define a function $N = 72(m+2)$ and observe that $N$ is polynomially
  bounded in $\abs{x}$.
  By Lemma~\ref{lemma:PP-eigenvalue-estimator}, there exists a language
  $B\in\class{PP}$ for which these implications hold for all
  $x\in\Sigma^{\ast}$ and $y_1,\ldots,y_N\in\Sigma^k$:
  \begin{align}    
    \lambda_{\text{min}}\biggl(\frac{S_{x,y_1} + \cdots + S_{x,y_N}}{N}\biggr)
    \geq \frac{2}{3}
    & \quad \Rightarrow \quad (x,y_1\cdots y_N) \in B,\\[2mm]
    \lambda_{\text{min}}\biggl(\frac{S_{x,y_1} + \cdots + S_{x,y_N}}{N}\biggr)
    \leq \frac{1}{3}
    & \quad \Rightarrow \quad (x,y_1\cdots y_N) \not\in B.
    \label{eq:MQRG(1) in QMA.PP no-implication}
  \end{align}

  Finally, for each input $x$, define a circuit $K_x$ that takes as input
  $N$ registers $(\reg{A}_1,\ldots,\reg{A}_N)$, each consisting of $n$ qubits,
  and outputs $N+1$ registers $(\reg{X},\reg{Y}_1,\ldots,\reg{Y}_N)$.
  The register $\reg{X}$ is initialized to the state $\ket{x}\bra{x}$, so that
  it simply echoes the input string $x$, and each register $\reg{Y}_j$ is
  obtained by independently applying the circuit $P_x$ to $\reg{A}_j$.
  Alternatively, one could write
  \begin{equation}
    K_x = \ket{x}\bra{x} \otimes P_x^{\otimes N},
  \end{equation}
  with the understanding that we are identifying the state
  $\ket{x}\bra{x}$ with the channel that inputs nothing and outputs the state
  $\ket{x}\bra{x}$.
  
  To prove that the promise problem $A$ is contained in
  $\class{QMA}\cdot\class{PP}$, it suffices to prove two things:
  \begin{trivlist}
  \item \emph{Completeness.}
    If it is the case that $x\in A_{\text{yes}}$, then there must exist a
    state $\xi\in\Density(\A^{\otimes N})$ such that
    \begin{equation}
      \op{Pr}(K_x(\xi)\in B) \geq \frac{2}{3}.
    \end{equation}
  \item \emph{Soundness.}
    If it is the case that $x\in A_{\text{no}}$, then for every state
    $\xi\in\Density(\A^{\otimes N})$ it must be that
    \begin{equation}
      \op{Pr}(K_x(\xi)\in B) \leq \frac{1}{3}.
    \end{equation}
  \end{trivlist}
  
  The proof of completeness follows a similar argument to the proof of
  Theorem~\ref{theorem:CQRG(1) in NP.PP}.
  Let $\rho\in\Density(\A)$ be any state for which \eqref{eq:MQRG-Alice-wins}
  is satisfied, and let $\xi = \rho^{\otimes N}$.
  It is evident that the output of $K_x(\xi)$ is given by $(x,y_1\cdots y_N)$,
  for $y_1,\ldots,y_N\in\Sigma^k$ sampled independently from the distribution
  \begin{equation}
    p(y) = \bra{y}P_x(\rho)\ket{y}.
  \end{equation}
  It follows by Lemma~\ref{lemma:flat-distribution-quality} that
  \begin{equation}
    \op{Pr}(K_x(\xi)\in B) \geq \frac{2}{3}.
  \end{equation}
  
  For the proof of soundness, the possibility that the state
  $\xi\in\Density(\A^{\otimes N})$ does not take product form must be
  considered.
  Our aim is to prove that if $y_1,\ldots,y_N$ are randomly selected according
  to the distribution that assigns the probability
  \begin{equation}
    \bigbra{y_1 \cdots y_N} P_x^{\otimes N} ( \xi) \bigket{y_1\cdots y_N}
  \end{equation}
  to each tuple $(y_1,\ldots,y_N)$, then
  \begin{equation}
    \op{Pr}\Biggl(
    \lambda_{\text{min}} \Biggl( \frac{S_{x,y_1} + \cdots + S_{x,y_N}}{N}
    \Biggr) \leq \frac{1}{3}\Biggr) \geq \frac{2}{3}\,,
  \end{equation}
  for this implies that $\op{Pr}(K_x(\xi)\in B) \leq 1/3$ by
  \eqref{eq:MQRG(1) in QMA.PP no-implication}.
  Toward this goal, choose a density operator $\sigma\in\Density(\B)$ for which
  \begin{equation}
    \label{eq:sigma bound}
    \sum_{y\in\Sigma^k}\bra{y}P_x(\rho)\ket{y} \tr\bigl( S_{x,y}\tsp
    \sigma\bigr) \leq \frac{1}{4}
  \end{equation}
  for all $\rho\in\Density(\A)$, which is possible by Sion's min-max theorem
  under the assumption \eqref{eq:MQRG-Bob-wins}, and define random variables
  $Z_1,\ldots,Z_N$ as
  \begin{equation}
    Z_j = \tr \bigl( S_{x,y_j} \tsp \sigma\bigr)
  \end{equation}
  for every $j\in\{1,\ldots,N\}$, assuming that $y_1,\ldots,y_N$ are chosen at
  random as above.
  It suffices to prove that
  \begin{equation}
    \label{eq:sum of Z random variables}
    \op{Pr}\Biggl(
    \frac{Z_1 + \cdots + Z_N}{N} 
      \leq \frac{1}{3}\Biggr) \geq \frac{2}{3}\,,
  \end{equation}
  as we have $\lambda_{\text{min}}(H) \leq \tr( H \tsp \sigma)$ for all
  Hermitian operators $H$.

  The complication we face at this point is that the random variables
  $Z_1,\ldots,Z_N$ are not necessarily independent (because $\xi$ does not
  necessarily have product form), so the most standard form of Hoeffding's
  inequality will not suffice to establish the required bound
  \eqref{eq:sum of Z random variables}.
  However, we observe that $Z_1,\ldots,Z_N$ are discrete random variables
  that take values in the interval $[0,1]$ and satisfy the inequality
  \begin{equation}
    \op{E}(Z_j | Z_1 = \alpha_1,\ldots,Z_{j-1} = \alpha_{j-1}) \leq 
    \frac{1}{4}
  \end{equation}
  for all $j\in\{2,\ldots,N\}$ and
  $\alpha_1,\ldots,\alpha_{j-1} \in [0,1]$ for which
  $\op{Pr}(Z_1 = \alpha_1,\ldots,Z_{j-1} = \alpha_{j-1})$ is nonzero.
  This is evident from the inequality \eqref{eq:sigma bound}, for it must hold
  when $\rho$ is equal to the reduced state of register $\reg{A}_j$,
  conditioned on any choice of $y_1,\ldots,y_{j-1}$ (and therefore on any
  choice of values $Z_1 = \alpha_1,\ldots,Z_{j-1} = \alpha_{j-1}$) that appear
  with nonzero probability.
  While the standard statement of Hoeffding's inequality does not suffice
  for our needs, the standard \emph{proof} of Hoeffding's inequality
  does establish that
  \begin{equation}
    \op{Pr}\Biggl(\frac{Z_1 + \cdots + Z_N}{N} \geq \frac{1}{3}\Biggr)
    =
    \op{Pr}\Biggl(\frac{Z_1 + \cdots + Z_N}{N} \geq \frac{1}{4}
    + \frac{1}{12}\Biggr)
    \leq
    \exp\biggl(-\frac{2N}{144}\biggr) < \frac{1}{3},
  \end{equation}
  as explained in an appendix at the end of the paper.
  Having obtained this bound, the proof is complete.
\end{proof}

\begin{cor}
  $\class{MQRG}(1) \subseteq \class{P}\cdot\class{PP}$.
\end{cor}

\section{Conclusion}

We have proved containments on two restricted versions of $\class{QRG(1)}$,
which we have called $\class{CQRG(1)}$ and $\class{MQRG(1)}$.
An obvious challenge is to prove a stronger containment on the class
$\class{QRG(1)}$ than $\class{PSPACE}$.
Observing that the containments we prove establish that
$\class{CQRG(1)}$ and $\class{MQRG(1)}$ are contained in the counting
hierarchy, we wonder whether $\class{QRG(1)}$ is also contained in the
counting hierarchy.

\appendix

\section{Hoeffding's inequality for dependent random variables with
  bounded conditional expectation}

In the proof of Theorem~\ref{theorem:MQRG(1) in QMA.PP} we used a slight
variant of Hoeffding's inequality, where the assumption of independence is
replaced by a bound on conditional expectation.
We expect that a bound along these lines has been observed before, but
we have not found a suitable reference.
(A similar bound is proved in \cite{BabaiCFLS1995} for Bernoulli random
variables, but we require the bound to hold more generally for discrete random
variables.)

It is, however, straightforward to adapt the most typical proof of
Hoeffding's inequality to obtain this bound, as we now explain.
We begin with Hoeffding's lemma, which is the essential ingredient in the
proof, and which we state without proof.
(A proof may be found in \cite{BhattacharyaW2016}, among many other
references.)

\begin{lemma}[Hoeffding's lemma]
  Let $X$ be a random variable taking values in $[\alpha,\beta]$,
  for real numbers $\alpha < \beta$, and assume $\op{E}(X) \leq 0$.
  For every $\lambda > 0$ it is the case that
  \begin{equation}
    \op{E}\bigl(\exp(\lambda X)\bigr)
    \leq \exp\biggl(\frac{\lambda^2}{8(\beta - \alpha)^2}\biggr).
  \end{equation}
\end{lemma}

\begin{remark}
  The more typical assumption for this lemma is that $\op{E}(X) = 0$,
  but (as is not surprising) it is true assuming instead that
  $\op{E}(X) \leq 0$.
  This follows immediately from the observation that if $\op{E}(X) \leq 0$,
  then
  \begin{equation}
    \op{E}(\exp(\lambda X)) \leq \op{E}(\exp(\lambda(X - \op{E}(X)))).
  \end{equation}
\end{remark}

The next lemma provides the inequality in the proof of Hoeffding's inequality
that would ordinarily follow from the assumption of independence.
For simplicity we prove this lemma for discrete random variables, which
suffices for our needs.

\begin{lemma}
  \label{lemma:conditional Hoeffding}
  Let $X$ and $Y$ be discrete random variables taking values in
  $[\alpha,\beta]$ for real numbers $\alpha < \beta$, and assume that
  $\op{E}(Y \,|\, X) \leq 0$.
  For every $\lambda > 0$ it is the case that
  \begin{equation}
    \op{E}(\exp(\lambda (X + Y))
    \leq
    \exp\biggl(\frac{\lambda^2}{8(\beta - \alpha)^2}\biggr)
    \op{E}(\exp(\lambda X)).
  \end{equation}  
\end{lemma}

\begin{proof}
  We may write
  \begin{equation}
    \op{E}(\exp(\lambda (X + Y))
    = \sum_x \exp(\lambda x) \op{E}(\exp(\lambda Y) \,|\, X = x)
    \op{Pr}(X = x),
  \end{equation}
  where the sum ranges over all possible values of $X$.
  By the assumption $\op{E}(Y \,|\, X) \leq 0$, Hoeffding's lemma implies
  \begin{equation}
    \begin{multlined}
      \sum_x \exp(\lambda x) \op{E}(\exp(\lambda Y) \,|\, X = x)
      \op{Pr}(X = x)\\
      \leq
      \exp\biggl(\frac{\lambda^2}{8(\beta - \alpha)^2}\biggr)
      \sum_x \exp(\lambda x) \op{Pr}(X = x)
      =
      \exp\biggl(\frac{\lambda^2}{8(\beta - \alpha)^2}\biggr)
      \op{E}(\exp(\lambda X)),
    \end{multlined}
  \end{equation}
  as required.
\end{proof}

Finally, we state and prove the variant of Hoeffding's inequality we have
used (again for discrete random variables).

\begin{theorem}
  Let $X_1,\ldots,X_n$ be discrete random variables taking values in $[0,1]$,
  let $\gamma\in[0,1]$, and assume that
  \begin{equation}
    \op{E}(X_k \,|\, X_1,\ldots,X_{k-1}) \leq \gamma
  \end{equation}
  for all $k\in\{1,\ldots,n\}$.
  For all $\varepsilon > 0$ it is the case that
  \begin{equation}
    \op{Pr}\bigl(X_1+\cdots+X_n \geq (\gamma + \varepsilon)n\bigr)
    \leq \exp(-2n\varepsilon^2).
  \end{equation}
\end{theorem}

\begin{proof}
  For every $\lambda > 0$ we have that
  \begin{equation}
    \begin{aligned}
      \op{Pr}\bigl(X_1+\cdots+X_n \geq (\gamma + \varepsilon)n\bigr)
      \hspace{-4cm} \\[1mm]
      & = \op{Pr}\bigl(\exp\bigl(
      \lambda (X_1 + \cdots + X_n - \gamma n)\bigr) \geq
      \exp(\lambda\varepsilon n)\bigr)\\[1mm]
      & \leq \frac{\op{E}\bigl(
        \exp\bigl(\lambda (X_1 + \cdots + X_n - \gamma n)\bigr)\bigr)}{
        \exp(\lambda\varepsilon n)}
    \end{aligned}
  \end{equation}
  by Markov's inequality.
  Applying Lemma~\ref{lemma:conditional Hoeffding} iteratively yields
  \begin{equation}
    \op{E}\bigl(\exp\bigl(\lambda (X_1 + \cdots + X_n - \gamma n)\bigr)\bigr)
    \leq 
    \exp\biggl(\frac{n\lambda^2}{8}\biggr).
  \end{equation}
  Choosing $\lambda = 4\varepsilon$ yields the claimed bound.
\end{proof}


\bibliographystyle{alpha}

\newcommand{\etalchar}[1]{$^{#1}$}

\end{document}